\pdfoutput=1
\documentclass[11pt]{article}
\usepackage[T1]{fontenc}
\usepackage{geometry} 
\usepackage[english]{babel}
\usepackage{cite}
\usepackage{amsmath}
\usepackage{amsthm}
\usepackage{amsfonts} % \mathbb
\usepackage{amssymb} % \nexists
\usepackage{xspace}
\usepackage{mathtools} % \DeclarePairedDelimiter
\usepackage{braket} % \set \Braket
\usepackage[linesnumbered,vlined]{algorithm2e}
\usepackage{xcolor} % \textcolor
\usepackage{footnote} 
\usepackage{environ} % \NewEnviron
\usepackage{suffix} % \WithSuffix
\usepackage{apptools}
\usepackage{hyperref}
\usepackage{float}
\usepackage[capitalise]{cleveref} % \cref
\usepackage{comment}

\usepackage[utf8]{inputenc}
\usepackage{graphicx}
\usepackage{svg}
\usepackage{algpseudocode} 

\makesavenoteenv{tabular}% footnotes in tables
\makesavenoteenv{table}% footnotes in tables

\makeatletter
\newtheorem*{rep@theorem}{\rep@title}
\newcommand{\newreptheorem}[2]{%
	\newenvironment{rep#1}[1]{%
		\def\rep@title{#2 \ref{##1}}%
		\begin{rep@theorem}}%
		{\end{rep@theorem}}}
\makeatother

\usepackage{thmtools} % should be after rep@theorem
\usepackage{thm-restate} % \begin{restatable} should be after rep@theorem

\newcommand{\size}[1]{\ensuremath{\left|#1\right|}}

\crefname{claim}{Claim}{Claims}
\crefname{property}{Property}{Properties}
\crefname{algocf}{Algorithm}{Algorithms}
\Crefname{algocf}{Algorithm}{Algorithms}

\newtheorem{theorem}{Theorem}[section]
\newreptheorem{theorem}{Theorem}
\newtheorem{lemma}[theorem]{Lemma}

\newtheorem{definition}{Definition}

\newtheorem{corollary}[theorem]{Corollary}

\newtheorem*{remark*}{Remark}

\theoremstyle{remark}

\theoremstyle{definition}

\DeclarePairedDelimiter{\ceil}{\lceil}{\rceil}
\newcommand{\ip}[1]{\left}

\newcommand{\hybrid}{\ensuremath{\mathsf{Hybrid}}\xspace}

\newcommand{\nccshort}{\ensuremath{\mathsf{NCC}}\xspace}
\newcommand{\ncczero}{\ensuremath{\nccshort_0}\xspace}
\newcommand{\congest}{\ensuremath{\mathsf{CONGEST}}\xspace}
\newcommand{\local}{\ensuremath{\mathsf{LOCAL}}\xspace}
\newcommand{\clique}{\ensuremath{\mathsf{Congested\ Clique}}\xspace}
\newcommand{\bcc}{\ensuremath{\mathsf{Broadcast\ Congested\ Clique}}\xspace}
\newcommand{\bccshort}{\ensuremath{\mathsf{BCC}}\xspace}

\newcommand{\cwc}{\ensuremath{\mathsf{CWC}}\xspace}
\newcommand{\hybridzero}{\ensuremath{\mathsf{Hybrid}_0}\xspace}

\newcommand{\kdis}{$k$\textsc{-dissemination}\xspace}
\newcommand{\kagg}{$k$\textsc{-aggregation}\xspace}

\newcommand\bigO[1]{\ensuremath{{O}(#1)}}
\WithSuffix\newcommand\bigO*[1]{\ensuremath{{O}\left(#1\right)}}
\newcommand\tildeBigO[1]{\ensuremath{{\tilde{{O}}}(#1)}}
\WithSuffix\newcommand\tildeBigO*[1]{\ensuremath{{\tilde{{O}}}\left(#1\right)}}
\newcommand\littleO[1]{\ensuremath{{o}(#1)}}
\WithSuffix\newcommand\littleO*[1]{\ensuremath{{o}\left(#1\right)}}
\newcommand\tildeLittleO[1]{\ensuremath{{\tilde{{o}}}(#1)}}
\WithSuffix\newcommand\tildeLittleO*[1]{\ensuremath{{\tilde{{o}}}\left(#1\right)}}
\newcommand\bigOmega[1]{\ensuremath{{\Omega}(#1)}}
\WithSuffix\newcommand\bigOmega*[1]{\ensuremath{{\Omega}\left(#1\right)}}
\newcommand\tildeBigOmega[1]{\ensuremath{{\tilde{{\Omega}}}(#1)}}
\WithSuffix\newcommand\tildeBigOmega*[1]{\ensuremath{{\tilde{{\Omega}}}\left(#1\right)}}
\newcommand\littleOmega[1]{\ensuremath{{\omega}(#1)}}
\WithSuffix\newcommand\littleOmega*[1]{\ensuremath{{\omega}\left(#1\right)}}
\newcommand\tildeLittleOmega[1]{\ensuremath{{\tilde{{\omega}}}(#1)}}
\WithSuffix\newcommand\tildeLittleOmega*[1]{\ensuremath{{\tilde{{\omega}}{\left(#1\right)}}}}
\newcommand\bigTheta[1]{\ensuremath{{\Theta}(#1)}}
\WithSuffix\newcommand\bigTheta*[1]{\ensuremath{{\Theta}\left(#1\right)}}
\newcommand\tildeTheta[1]{\ensuremath{{\tilde{{\Theta}}}(#1)}}
\WithSuffix\newcommand\tildeTheta*[1]{\ensuremath{{\tilde{{\Theta}}\left(#1\right)}}}

\makeatletter
\newcommand*{\whp}{%
    \@ifnextchar{.}%
        {w.h.p}%
        {w.h.p.\@\xspace}%
}
\makeatother

\RestyleAlgo{boxruled}
\LinesNumbered
\let\oldnl\nl% Store \nl in \oldnl
\newcommand{\nonl}{\renewcommand{\nl}{\let\nl\oldnl}}% Remove line number for one line

\SetKwFunction{AssignHelpers}{Assign-Helpers}
\SetKwFunction{TransmitSkeleton}{Transmit-Skeleton}
\SetKwFunction{LearnNeighborhood}{Learn-Neighborhood}
\SetKwFunction{LearnHelpers}{Learn-Helpers}
\SetKwFunction{ReassignSkeletons}{Reassign-Skeletons}

\NewEnviron{killcontents}{}
\newcommand{\tk}[1][k]{T_{#1}}
\newcommand{\tn}[1][n]{T_{#1}}

\newcommand{\hop}{\text{hop}}
\DeclareMathOperator*{\argmax}{arg\,max}

\title{Universally Optimal Deterministic Broadcasting \\ in the HYBRID Distributed Model}
\author{
Yi-Jun Chang 
\\
National University of Singapore
\\
\small{\texttt{cyijun@nus.edu.sg}}
\and
Oren Hecht
\\
Technion
\\
\small{\texttt{hecht.oren@campus.technion.ac.il}}
\and
Dean Leitersdorf
\\
National University of Singapore
\\
\small{\texttt{dean.leitersdorf@gmail.com}}
}

\date{}

\begin{document}
\begin{titlepage}
\maketitle
\thispagestyle{empty}
\begin{abstract}
In theoretical computer science, it is a common practice to show \emph{existential} lower bounds for problems, meaning there is a family of pathological inputs on which no algorithm can do better than the lower bound. However, in many cases most inputs of interest can be solved much more efficiently, giving rise to the notion of \emph{universal optimality}. Roughly speaking, a universally optimal algorithm is one that, given some input, runs as fast as the best algorithm designed specifically for that input.

Questions on the existence of universally optimal algorithms in distributed settings were first raised by Garay, Kutten, and Peleg in FOCS '93. This research direction reemerged recently through a series of works, including the influential work of Haeupler, Wajc, and Zuzic in STOC '21, which resolves some of these decades-old questions in the supported \congest model.

We work in the \hybrid distributed model, which 
analyzes networks combining both global and local communication. Much attention has recently been devoted to solving distance related problems, such as All-Pairs Shortest Paths (APSP) in \hybrid, culminating in a $\tilde \Theta(n^{1/2})$ round algorithm for exact APSP. However, by definition, every problem in \hybrid is solvable in $D$ rounds, where $D$ is the diameter of the graph, showing that while $\tilde \Theta(n^{1/2})$ rounds is existentially optimal for APSP, it is far from universally optimal.

We show the \emph{first} universally optimal algorithms in \hybrid, by presenting a fundamental tool that solves \emph{any} broadcasting problem in a universally optimal number of rounds, deterministically. Specifically, we consider the \kdis problem, which given an $n$-node graph $G$ and a set of $k$ messages $M$ distributed arbitrarily across $G$, requires every node to learn all of $M$. We show a universal lower bound and a matching, deterministic upper bound, for any graph $G$, any value $k$, and any distribution of $M$ across $G$.

This broadcasting tool opens a new exciting direction of research into showing universally optimal algorithms in \hybrid. As an example, we use it to obtain algorithms to approximate APSP in general graphs and to solve APSP exactly in sparse graphs; these algorithms are universally optimal in that they match the lower bound for even just for learning the, potentially random, identifiers of the nodes in the graph, which are needed for outputting shortest path distances.

\end{abstract}
\end{titlepage}

\tableofcontents

\newpage

\section{Introduction}
We work in \hybrid, a key model of distributed computation, and tackle the fundamental problem of broadcasting information over a graph -- deterministically solving the most general variant of this problem in the best possible complexity. The \hybrid model of distributed computing abstracts common practical distributed networks in order to provide a framework for performing theoretical research, which can be readily adapted to uses in modern data centers and distributed networks \cite{Farrington2010, Halperin2011, Tell2018SDWANAM, ROSSBERG20111684}. In any distributed network, broadcasting information is a fundamental task, which is interesting either on its own as an end goal (e.g.,~to broadcast a network update, notification of failure, etc.) or as a basic building block for solving other problems (e.g.,~for computing paths and distances between nodes). 

We solve the most general version of broadcasting, whereby there are some $k$ messages, for \emph{any} value $k$, originally distributed in \emph{any} fashion across the graph (i.e.,~all messages can begin at one node, or might be spread out such that each node holds one message, etc.), and it is desired that every node in the graph learns all $k$ messages. We solve this problem in a \emph{universally optimal} way, implying that our algorithm is as fast as the best possible algorithm which even knows the graph topology ahead of time and the original locations (but not contents) of the messages. In essence, we design one general algorithm, which works on any graph and with any original message distribution, and it is impossible to show a faster algorithm, including algorithms tailor-made for specific graphs and original message distributions. Finally, we use our broadcasting tool to also approximate distances and cuts in the graph.

\paragraph{Hybrid Networks.} The \hybrid model \cite{kuhn2020computing} investigates distributed networks whereby nodes physically close to each other can communicate via high-bandwidth local communication links, while there are also low-bandwidth global communication links to send small amounts of information between physically distant parts of the network. 
These types of hybrid networks appear in real-world applications, including data centers with limited wireless communication, and high-bandwidth short-ranged wired communication \cite{Halperin2011, Farrington2010}. Additional examples include cellular networks where devices can communicate in their local environment with a high bandwidth link (e.g.,~communication between nearby smartphones using Bluetooth, WiFi Direct, or LTE Direct), in addition to global communication through a lower-bandwidth cellular infrastructure \cite{Kar2018}.

The theoretical research of hybrid networks via the \hybrid model so far mainly explored distance computation tasks, such as the $k$-Source Shortest Path and specifically the Single Source Shortest Path (SSSP, $k=1$) and All-Pairs Shortest Path (APSP, $k=n$), diameter computation, and more \cite{censor2020distance, censor2021sparsity, anagnostides2021deterministic, Coy2021Near, coy2022routing, Kuhn2022Routing, Feldmann2020Fast}. The key observation is that the combination of 
both a high-bandwidth local network and a low-bandwidth global network allows solving problems significantly faster than is possible in either network alone -- for instance, APSP requires $\tildeBigOmega{n}$ rounds\footnote{The $\tildeBigO{\cdot}, \tildeBigOmega{\cdot}, \tildeTheta{\cdot}$ notation hides polylogarithmic factors.} in either the local or global network alone \cite{augustine2020shortest}, yet can be solved in $\tildeTheta{\sqrt{n}}$ rounds in \hybrid by using both networks together \cite{augustine2020shortest, kuhn2020computing}. Recently, a new line of work started investigating the use of \emph{routing schemes} and \emph{distance oracles} \cite{Kuhn2022Routing, Coy2021Near}. These are fundamental tools for applications like efficient packet-forwarding, which stand in the backbone of the modern-day internet.

\paragraph{Broadcasting in \hybrid.} A fundamental use case for hybrid networks is the broadcasting of information. Broadcasting information is interesting on its own as an end goal, such as announcing a failure, a change of policy, or other control messages in a data center. Further, broadcasting itself can be a tool useful for solving other problems, as evident by the entire research field around the \bcc model (\bccshort) \cite{drucker2014on} -- in \bccshort, in each round every node can broadcast one message to every other node, and using only this basic primitive it was shown that many problems can be solved \cite{chen2019broadcast, becker2015hierarchy, jurdzinski2018connectivity, becker2020impact, montealegre2016brief, holzer2015approximation}.

We investigate the most general broadcast variant, whereby there are some $k$ messages spread out arbitrarily across the graph, and it is desired for all these messages to be known to all the graph. The messages originally can start in any configuration -- i.e.,~all from one specific node, or spread out across different nodes in the graph. We solve this problem for \emph{any} value $k$ and any original message distribution. As a corollary, by setting $k=n$, we show a simulation of \bccshort in \hybrid.

\paragraph{Universal Optimality.} All research in \hybrid so far focused on \emph{existential} lower bounds, meaning there is a pathological graph family where no algorithm can do better than some lower bound. For instance, showing that APSP has a lower and upper bound of $\tildeTheta{\sqrt{n}}$ \cite{augustine2020shortest, kuhn2020computing}, by showing a family of graphs that require $\tildeBigOmega{\sqrt{n}}$ rounds to solve APSP in, and a matching upper bound. However, trivially, any problem in \hybrid can be solved in $D$ rounds, where $D$ is the network diameter, just using the high-bandwidth local network. In practical examples, many networks have a small diameter compared to the number of nodes $n$, rendering the state-of-the-art existentially optimal algorithms impractical. As a further example, it was shown in \cite{anagnostides2021deterministic, Feldmann2020Fast}, that there are graphs where one can solve APSP exponentially faster, in just $O(\log n)$ rounds, and many networks of interest can offer drastically faster algorithms compared to the existential lower bounds.

Therefore, a worthy goal is universally optimal algorithms, a concept that was first theorized in the distributed setting by \cite{Garay1998} in FOCS '93. Loosely speaking, a universally optimal algorithm runs as fast as possible on \emph{any} graph, not just worst case graphs (a formal definition soon follows). In \cite{ghaffari2016distributed} the first steps towards non-worst case algorithms were taken, in the well-known \congest distributed model, with the introduction of the low-congestion shortcut framework. This was followed by a line of influential works \cite{Haeupler2021,ghaffari_dfs,ghaffari_mixing_time,haeupler2016near, haeupler2018faster, haeupler2018minor, haeupler2022hop, kitamura2021low, rozhovn2022undirected, zuzic2018towards, Ghaffari2022, zuzic2022universally,ghafarriOPTBroadcast}, and culminating in the definition of universal optimality in the work of \cite{haeupler2021universally}.

We show the \emph{first universally optimal algorithm in \hybrid}. We present a parameter $\tk(G)$, for any graph $G$, and show it is a universal lower bound for broadcasting $k$ tokens to all of $G$. Namely, on the specific graph $G$, no algorithm can solve broadcast $k$ messages in less than $\tildeBigOmega{\tk(G)}$ rounds, even if it knows the entire topology of $G$ and the initial locations (but not contents) of the $k$ messages. We complement this lower bound with a \emph{single}, deterministic algorithm that, when it runs on any graph $G$, takes $\tildeBigO{\tk(G)}$ rounds and solves the broadcast problem. We stress that the complexity of our algorithm does not depend on the original token distribution.

~\\
In essence, universally optimal algorithms are an important step in distributed research, both presenting a theoretical challenge, and \emph{bridging a gap between theory and practice} by showing algorithms that are optimal for any specific case, including real-world graphs. We believe that setting the foundations in \hybrid for such research opens the doors to much further exciting results to come.

\paragraph{Roadmap.} We now proceed to an overview of our contributions and techniques developed to show them. In \cref{sec:optimal_broadcast} we show our universal lower bound and matching upper bound for broadcast, and also discuss a similar algorithm for aggregating $k$ functions. In \cref{sec:applications} we utilize our broadcasting tool to approximate APSP and various cuts. Finally, in \cref{sec:graph_families} we compute bounds on $\tk(G)$ in certain graph families, to give a taste as to how $\tk(G)$ relates to other graph parameters such as $n$ and $D$. This also shows that the \emph{universally} optimal algorithm improves significantly over the existing state-of-the-art, \emph{existentially} optimal algorithms, in such graph families.

\subsection{Our Contributions}

We now describe our results and briefly explain the techniques behind our algorithms.

Before we proceed, we provide a rough definition of \hybrid (formal definition to follow later). We are given an initial input graph $G = (V, E)$, and proceed in synchronous time steps called rounds. In each round, any two nodes in $V$ with an edge in $E$ between them can communicate any number of bits, through the \emph{local network}. Further, every node $v \in V$ can choose $\log n$ nodes arbitrarily in $V$ and send them each a (possibly unique) $O(\log n)$-bit message through the \emph{global network}. Every node in $V$ can be the target of only $\log n$ messages via the global network, per round. 

A more restrictive variant of \hybrid is \hybridzero (see \ncczero in \cite{augustine2021distributed}), in which every node has some arbitrary $O(\log n)$ bit identifier, every node originally only knows the identifiers of itself and its neighbors, and knows nothing about the identifiers of the other nodes. The key difference between \hybrid and \hybridzero is that in \hybridzero a node must first learn the identifier of another node before it is able to send messages to that node over the global network.

Clearly, if an algorithm works in \hybridzero, then it works in \hybrid too, and if a lower bound holds in \hybrid, then it holds in \hybridzero too.

\subsubsection{Universally Optimal Broadcasting}

Our main research question is the following broadcast problem.

\begin{definition}[\kdis]
\label{def:kdis}
    Given any set of messages $M$, where $k = |M|$ and each $m \in M$ is originally known to only one node in the graph, the \kdis problem requires that all messages $M$ become known to every node in the graph.
\end{definition}

A similar variant of token dissemination is presented as one of the most basic \hybrid communication primitives, in the paper defining the model \cite{augustine2020shortest}. There, they limit the number of tokens originally at any node by some value $\ell$ and provide a randomized algorithm operating in $\tildeBigO{\sqrt{k}+\ell}$ rounds. Recently, \cite{anagnostides2021deterministic} removed the limitation of $\ell$ by showing a deterministic algorithm operating in $\tildeBigO{\sqrt{k}}$ rounds, however, they require $k\geq{}n$. In our case, we solve the most general variant of the problem, with no bounds on $k$ or on the original distribution of the tokens in the graph.

We begin by defining the \emph{broadcast quality} of a graph. For any node $v$, denote by $B_t(v)$ the ball of radius $t$ around $v$ -- that is, all nodes which can reach $v$ with a path of at most $t$ edges.

\begin{definition}[Broadcast Quality]
\label{def:bcast_quality}
    Given a graph $G = (V, E)$, value $k$, and a node $v$, let
    \[
    \tk(v) = \min \{\{t \mid \size{B_t(v)} \geq k/t\} \cup{} \{D\}\}
    \]
    and let $\tk(G) = \max_{v \in V} \tk(v)$.
\end{definition}

When $G$ is clear from context, we write $\tk$ instead of $\tk(G)$. Notice that $\tk$ actually characterizes a property of the set of \emph{power graphs} of $G$. A power graph $G^t$ of $G$ has the same node set as $G$, and has an edge $e = \{v, u\}$ if there is a path between $v$ and $u$ in $G$ with at most $t$ edges. $\tk$ is essentially the minimal value $t = \tk$ such that the minimum degree in $G^t$ is at least $k/t$.

We now proceed to showing lower and upper bounds of $\tilde\Theta(\tk)$ for solving \kdis in \hybrid. This shows a couple of interesting properties. First, it shows that in \hybrid, the complexity of \kdis depends \emph{only on the graph topology}, and not on the original distribution of the tokens to broadcast. Moreover, it shows that the minimal degree of the set of power graphs is a fundamental property in \hybrid. 

Intuitively, this makes sense as in \hybrid within $t$ rounds every node in $G$ can communicate with every neighbor it has in $G^t$. At the same time, in those $t$ rounds each of those nodes can receive $O(t \cdot \log n)$ messages through the global network. Thus the minimum degree of a node in $G^t$ dictates an upper bound on the amount of information that node can receive in $t$ rounds by combining the global and local networks. In comparison to the \local model where nodes communicate only with their neighborhoods and in $t$ rounds every node knows only the information originally stored in each of its neighbors in $G^t$, in \hybrid we have the added benefit of choosing which messages are routed in the global network, and then having the nodes use the local network to receive also the messages their nearby nodes saw over the global network.

\paragraph{Lower Bound.} We first show that it takes $\tilde\Omega(\tk)$ rounds in \hybrid to solve \kdis.

\begin{restatable}{theorem}{kdisLB}
\label{thm:kdisLB}
There is a universal lower bound of $\tildeBigOmega{\tk}$ rounds for solving the \kdis problem in \hybrid.
\end{restatable}

A rough outline of the proof is as follows. 
We notice that it is possible to assume that messages can only be routed \emph{as is}, i.e.,~without any coding techniques to shorten messages or compute a shorter representation of a subset of the $k$ messages. The idea behind this step is that we claim universal optimality w.r.t.~the graph and the locations of the messages, but no w.r.t.~the contents of the messages, and thus in the worst case the messages can be random bits and so any compression of the messages, which still works \whp,\footnote{With high probability (\whp) means that for an arbitrary but constant $c>0$, the probability of success is at least $1 - 1/n^c$.} reduces the required number of rounds by at most a constant factor. Thus, assuming that messages are only routed as-is implies at most a constant factor slowdown to the round complexity.

Once it is established that messages can only be routed as-is, observe $v^*$, the node in $G$ where $\tk(v^*) = \tk$. Assume that there is a message $m$ which is originally located at a node at $u$ which is $t$ hops from $v^*$. In order for $v^*$ to learn the message in $t' < t$ rounds, it must be the case that $m$ was at some point sent across a global edge to some node in $B_{t'}(v^*)$. Assume for the sake of contradiction that this is not the case -- i.e.,~that in the first $t'$ rounds of the algorithm, $m$ traveled only via local edges or sent via global edges to nodes in $V \setminus B_{t'}(v^*)$. This would imply that $m$ would have made its way from some node $w \in (V \setminus B_{t'}(v^*))$ to $v^*$ using only local edges, which is clearly impossible in $t'$ rounds.

Thus, we look at $B_{\tk - 1}(v^*)$. By \cref{def:bcast_quality},  $B_{\tk}(v^*)$ is the smallest-radius ball around $v^*$ such that $|B_{\tk}(v^*)| \geq k/\tk$, and so it holds that $|B_{\tk-1}(v^*)| < k/(\tk-1)$. If at least $k/2$ of the messages are originally outside $B_{\tk-1}(v^*)$, then we show that node $v^*$ cannot learn all the $k$ messages in the graph in $\tk/(2\log n)$ rounds. Denote the $k/2$ messages originally outside of $B_{\tk-1}(v^*)$ by $M'$. In order for $v^*$ to learn all the messages in $(\tk-1)/(2\log n) < \tk - 1$ rounds, then each message in $M'$ must at some point travel through a global edge to some node in $B_{\tk-1}(v^*)$. However, as $|B_{\tk-1}(v^*)| < k/(\tk-1)$, then in $(\tk-1)/(2\log n)$, rounds, these nodes can receive only $(\tk-1)/(2\log n) \cdot |B_{\tk-1}(v^*)| \cdot \log n < k/2$ messages, as each node can only receive $\log n$ global messages per round, due to the definition of \hybrid. As $|M'| \geq k/2$, this implies a contradiction.

Conversely, if at least $k/2$ message are actually originally inside $B_{\tk-1}(v^*)$, then we show that $B_{\tk-1}(v^*)$ does not have the capacity to \emph{send} these messages out of $B_{\tk-1}(v^*)$ to the rest of the graph, i.e.,~to nodes in $V \setminus B_{\tk-1}(v^*)$.

\paragraph{Upper Bound.} We compliment this lower bound by showing that the same round complexity suffices to solve \kdis, deterministically.

\begin{restatable}{theorem}{kdisUP}
\label{thm:optimal_dissemination}
The \kdis problem can be solved
    \emph{deterministically} 
    in $\tildeBigO{\tk}$ rounds in \hybrid. This result even holds in the more restrictive \hybridzero.
\end{restatable}

To show \cref{thm:optimal_dissemination}, we partition the graph into \emph{clusters}, each with diameter $\tildeBigO{\tk}$ and with $\Theta(k/\tk)$ nodes. We desire to ensure that all the $k$ messages arrive at each cluster, which will allow each node to ultimately learn all the $k$ messages by using the local edges to receive all the information its cluster has.

To do so, we begin by building a virtual tree of all the clusters. While trivial in the \hybrid model, it is rather challenging in \hybridzero, as we must construct a tree that spans the entire graph, has $\tildeBigO{1}$ depth and constant degree, and every two nodes in the tree know the identifiers of each other even though they may be distant in the original graph. To do so, we build upon certain overlay construction techniques from \cite{gmyr2017distributed}.

Then, between any parent $P$ and child $C$ clusters in the binary tree, we ensure that every node node $P$ knows the identifier of exactly one node in $C$ and vice-versa, so that they may communicate through the global network. Reaching this state requires great care, as sending the identifiers of all the nodes in one cluster to all the nodes in another cluster is a challenging task. This must be done over the global network (as the clusters might be physically distant from each other) and thus we must carefully design an algorithm to achieve this without causing congestion.

Once nodes in clusters $P$ and $C$ can communicate with each other, we propagate all the $k$ messages up the cluster tree so that the cluster at the root of the tree knows all $k$ messages. Then, we propagate them back down to ensure that every cluster receives the $k$ messages. While doing these propagations, we perform load balancing steps within each cluster in order to ensure that each of its nodes is responsible to send roughly the same number of messages to other clusters, preventing congestion in the global network.

\paragraph{Corollaries.} An immediate corollary of the above is a \bccshort simulation in \hybrid, which requires broadcasting $k=n$ tokens spread uniformly across $G$. Therefore, we use $\tn$, as follows.

\begin{corollary}
    There is a universal lower bound of $\tildeBigOmega{\tn}$ rounds for simulating one round of \bccshort in \hybrid, and there exists a deterministic algorithm which does so in $\tildeBigO{\tn}$ rounds in \hybrid (and even in \hybridzero).
\end{corollary}

As another corollary to our \kdis bounds, we get the same universally optimal characterization for the \kagg problem.

\begin{definition}[\kagg]
    Let $F: X\times X \rightarrow X$ be an aggregation function (associative and commutative). Assume each node $v$ originally holds $k$ values $f_1(v), \dots, f_k(v)$. The \kagg problem requires that each node learn all the values $F(f_i(v_1),\dots,f_i(v_n))$, for every $i \in [k]$. It is assumed that $k$ and the values $f_1(v), \dots, f_k(v)$ for each $v \in V$ are all at most polynomial in $n$.
\end{definition}

The idea behind our this is showing a bidirectional reduction in \hybrid between \kdis and \kagg, implying that both the lower and upper bounds above transfer. Showing both directions of the reduction requires some technical work. Showing that \kagg solves \kdis requires a routine to coordinate between all the nodes in the graph, and showing that \kdis solves \kagg requires specific observations about the trees we construct in the algorithm in \cref{thm:optimal_dissemination}.
We thus get the following.

\begin{restatable}{theorem}{kaggLBUB}
\label{thm:kaggLBUB}
    There is a universal lower bound of $\tildeBigOmega{\tk}$ rounds in \hybrid for the \kagg problem, and there is a deterministic algorithm that solves it in $\tildeBigO{\tk}$ rounds in \hybrid (and even in \hybridzero).
\end{restatable}

\subsubsection{Analysis of \texorpdfstring{$\tk$}{Tk}}

We analyze the properties of $\tk$ in order to compare it to other graph parameters, such as the number of nodes in a graph and its diameter. This allows us to compare our results to existing previous works.

\begin{restatable}{lemma}{boundingTk}
\label{lem:boundingTk}
    If $\tk\neq{}D$, then $\sqrt{\frac{D}{3n}k}\leq{}\tk\leq{}\min{}\{D,\sqrt{k}\}$.
\end{restatable}

Recall that the previous works for solving weaker variants of \kdis all take $\tildeBigO{\sqrt{k}}$ rounds \cite{augustine2020shortest, anagnostides2021deterministic}, and so \cref{lem:boundingTk} implies that our algorithm for \kdis in $\tildeBigO{\tk}$ rounds is never slower than the previous works and supports a wider variety of cases. 

Conversely, we analyze certain graphs of families where $\tk = o(\sqrt{k})$, to show that many such graphs exist, beyond just graphs with $D = o(\sqrt{k})$. We proceed with estimating the value of $\tk$ for path, cycle, and any $d$-dimensional square grid graphs (formal definition to follow). 
For instance, in $d$-dimensional square grids, we get the following.

\begin{restatable}{theorem}{gridGraphsTk}
\label{theorem:grid_graphs_tk}
    Let $G=(V,E)$ be a $d$-dimensional square grid, with $|V|=n=m^d$. \[
    \tk=\begin{cases}
        \Theta{}(k^{1/(d+1)}) & k=O(n^{(d+1)/d}) \\
        O(D)=O(dn^{1/d}) & k=\Omega(n^{(d+1)/d})
    \end{cases}
    \]
\end{restatable}

This shows a few interesting points. First, when $k$ is not too big, i.e.,~$k = O(n^{1+1/d})$, we can broadcast $k$ messages in a much better round complexity, $\tildeBigO{\tk}=\tildeBigO{k^{1/(d+1)}}$ than the existing algorithms, which take at least $\tildeBigO{\sqrt{k}}$. However, once $k$ crosses $\Omega{}(n^{1+1/d})=\Omega{}(D^{d+1})$, we cannot do much better than naively sending all messages via the local network in $O(D)$ rounds.

Specifically, note that for any constant dimension $d$, grid graphs have $O(dn) = O(n)$ edges, and thus in $\tildeBigO{\tn} = O(n^{1/(d+1)})$ rounds, it is possible for all the nodes to learn the entire graph, using \cref{thm:optimal_dissemination}, and locally compute exact APSP. For any $d\geq 2$, this is polynomially faster than the existentially optimal algorithms of \cite{augustine2020shortest, kuhn2020computing, anagnostides2021deterministic}.

\subsubsection{Applications}

We show a variety of applications for our \kdis and \kagg results. Most of these applications follow from our above algorithms in a rather straightforward manner, as broadcasting and aggregation are very fundamental building blocks.

We start with several applications for APSP in \hybrid. In previous works, exact weighted APSP was settled with an existential lower and upper bounds of $\tildeBigO{\sqrt{n}}$ rounds, due to \cite{augustine2020shortest, kuhn2020computing}. Their algorithms are randomized, while the best known deterministic algorithm of \cite{anagnostides2021deterministic} runs in $\tildeTheta{\sqrt{n}}$ rounds, yet produces a $(\log{n}/\log{\log{n}})$-approximation. For unweighted APSP, \cite{anagnostides2021deterministic} showed a $(1+\epsilon)$-approximation, deterministically, in $\tildeBigO{\sqrt{n}}$ rounds as well.

We show several algorithms for exactly computing or approximating APSP. Before we show our upper bounds, we stress that they all work in \hybridzero. Note that in this setting, the identifiers can be arbitrary $O(\log n)$ strings. In order for a node to produce its output for APSP, it must know the identifiers of all the nodes in $G$, and thus this corresponds to broadcasting all $n$ identifiers -- which are, in essence, $n$ arbitrary messages. Therefore, the following holds due to \cref{thm:kdisLB}.\footnote{Note that in \kdis, we assume that each of the $k$ messages originally is only known to one node. In this setting where we have to broadcast the $n$ arbitrary identifiers of the nodes in $G$, it actually holds that the identifier of each node is known both to itself and its neighbors. This is not a problem as one can assume w.l.o.g.~that any algorithm for \kdis can perform one round \emph{for free} whereby each node originally holding messages sends these messages to all its neighbors using the local network. This simply implies that the lower bound in \cref{thm:apsp_lb} is lower by at most one round than the lower bound in \cref{thm:kdisLB}.}

\begin{theorem}
\label{thm:apsp_lb}
    In \hybridzero, there is a universally optimal lower bound of $\tildeBigOmega{\tn}$ rounds for any approximation of APSP.
\end{theorem}

Note that the universal optimality of the lower bound is w.r.t.~the graph topology, but not w.r.t.~the choice of identifiers themselves.

Most of our algorithms below run in $\tildeBigO{\tn}$ rounds, and therefore are universally optimal computations of APSP in \hybridzero. Moreover, due to \cref{lem:boundingTk}, $\tn \leq \sqrt{n}$ for any graph, implying that our algorithms run in $\tildeBigO{\tn} = \tildeBigO{\sqrt{n}}$ rounds, and thus are never slower than the $\tildeTheta{\sqrt{n}}$ round algorithms of \cite{augustine2020shortest, kuhn2020computing} in \hybrid, yet are faster when $\tn = o(\sqrt{n})$.

As a warm up, we show that in sparse graphs one can learn the entire graph and thus exactly compute APSP. That is, in graphs with $\tildeBigO{n}$ edges, we apply \cref{thm:optimal_dissemination} with $k=n$ at most $\tildeBigO{1}$ times, resulting in $\tildeBigO{\tn}$ rounds.

\begin{corollary}
    Given a sparse, weighted graph $G=(V,E,\omega)$ with $|E|=\tildeBigO{n}$, there is an algorithm that solves any graph problem in $\tildeBigO{\tn}$ rounds in \hybridzero, including exact weighted APSP.
\end{corollary}

We proceed with approximating APSP in graphs with any number of edges.

\paragraph{Unweighted APSP.} We show a $(1+\epsilon)$ approximation of unweighted APSP which runs in $\tildeBigO{\tn/\epsilon^2}$ rounds \whp. The best known algorithm for exactly computing or approximating APSP for general graphs to practically any factor takes $\tildeBigO{\sqrt{n}}$ rounds \cite{kuhn2020computing}. As $\tn \leq \sqrt{n}$, we are always at least as fast, and are faster when $\tn = o(\sqrt{n})$.

\begin{restatable}{theorem}{unweightedAPSPApprox}
\label{thm:unweighted_apsp_approx}
    For any $\epsilon \in (0, 1)$, there is a randomized algorithm which computes a $(1+\epsilon)$-approximation of APSP in unweighted graphs in \hybridzero, in $\widetilde{O}(\tn/\epsilon^2)$ rounds \whp.
\end{restatable}

Basically, we extend the the $(1+\epsilon)$ polylogarithmic SSSP algorithm of \cite{schneider2023thesis} to a universally optimal unweighted APSP.
The techniques behind this algorithm are as follows. We first compute a $\tildeBigO{\tn}$-weak diameter clustering to at most $\tn$ clusters. Then, we execute the $(1+\epsilon)$ polylogarithmic SSSP from cluster centers. We then explore a large enough neighborhood of each node, and each node broadcasts its closest cluster center and distance to it. Finally, we are able to approximate the distance well enough using the $(1+\epsilon)$ approximation to cluster center and the distances broadcast. 

\paragraph{Weighted APSP.} We now show several results for approximating APSP in weighted graphs.

We begin with the following algorithm that computes an $(\epsilon \cdot \log n)$ approximation, and is based on using a known result of \cite{rozhovn2020polylogarithmic} for constructing a graph spanner (a sparse subgraph which preserves approximations of distances) and then learning that entire spanner.

\begin{restatable}{theorem}{firstWeightedAPSPApprox}
\label{thm:firstWeightedAPSPApprox}
    There is a deterministic algorithm in \hybridzero, that given a graph $G=(V,E,\omega)$, computes a $(\epsilon \cdot \log{n})$-approximation for APSP in $\widetilde{O}(2^{1/\epsilon}\tn)$ rounds.
\end{restatable}

We can use \cref{thm:firstWeightedAPSPApprox} to show a result which is comparable to the best known deterministic approximation, by \cite{anagnostides2021deterministic}, deterministically achieving the same approximation ratio of $\log{n}/\log{\log{n}}$, but in $\tildeBigO{\tn}$ instead of $\tildeBigO{\sqrt{n}}$ rounds.

\begin{corollary}
    By running \cref{thm:firstWeightedAPSPApprox} with $\epsilon=1/\log{\log{n}}$, we achieve a $\log{n}/\log{\log{n}}$ approximation in $\widetilde{O}(\tn)$ rounds in \hybridzero.
\end{corollary}

Finally, we show the following result for approximating APSP. This runs slightly slower than $\tildeBigO{\tn}$ rounds, yet, shows a much better approximation ratio. For instance, for a $3$-approximation of weighted APSP, it achieves a round complexity of $\tildeBigO{n^{1/4}\tn^{1/2}}$, which is always less than $\tildeBigO{\sqrt{n}}$, as $\tn \leq \sqrt{n}$. This result is based on the well-known skeleton graphs technique,  first observed by \cite{ullman1990high}.

\begin{restatable}{theorem}{secondWeightedAPSPApprox}
\label{theorem:secondWeightedAPSPApprox}
    For any integer $\alpha \geq 1$, there is a randomized algorithm that computes a $(4\alpha - 1)$-approximation for APSP in weighted graphs in $\tildeBigO{\alpha \cdot n^{1/(3\alpha + 1)} (\tn)^{2/(3 + 1/\alpha)} + \alpha \cdot \tn}$ rounds in \hybridzero, \whp.
\end{restatable}

\subsubsection{Approximating Cuts}
Similarly to \cite{anagnostides2021deterministic}, we leverage cut-sparsifiers \cite{spielman2004nearly} and their efficient implementation in \congest \cite{koutis2016simple} to approximate any cut and solve several cut problems. Our algorithms runs in $\tildeBigO{\tn/\epsilon + 1/\epsilon^2}$ rounds in \hybridzero, which is always at least as fast as all current algorithms, and faster when $\tn = o(\sqrt{n})$. The idea behind our result is to execute the \congest algorithm of \cite{koutis2016simple}, in $\tildeBigO{1/\epsilon^2}$ rounds, to create a subgraph with $\tildeBigO{n/\epsilon^2}$ edges which approximates all cuts. Then, we broadcast that subgraph in $\tildeBigO{\tk[n/\epsilon^2]} = \tildeBigO{\tn/\epsilon}$ rounds using \cref{thm:optimal_dissemination}.

\begin{restatable}{theorem}{minCutApprox}
\label{theorem:minCutApprox}
    For any $\epsilon \in (0, 1)$, there is an algorithm in \hybridzero that runs in $\tildeBigO{\tn/\epsilon + 1/\epsilon^2}$ rounds, \whp, after which each node can locally compute a $(1+\epsilon^2)$-approximation to all cuts in the graph. This provides approximations for many problems including minimum cut, minimum $s$-$t$ cut, sparsest cut, and maximum cut.
\end{restatable}

\subsection{Further Related Work}

\paragraph{\hybrid.} The \hybrid model in its current form was recently introduced in \cite{augustine2020shortest}. Since then, most research focused on shortest paths computations and closely related problems such as diameter calculation. In \cite{augustine2020shortest} there is an existential lower bound of $\tildeBigOmega{\sqrt{n}}$ rounds for APSP, even for $O(\sqrt{n})$-approximations. This was generalized by \cite{kuhn2020computing}, to $\tildeBigOmega{\sqrt{k}}$ rounds for k-SSP. In \cite{kuhn2020computing} an existentially optimal, randomized $\tildeBigO{\sqrt{n}}$ algorithm for exact APSP is shown. Both papers make heavy use of token dissemination, token routing, skeleton graphs and simulating \clique (a different distributed model) algorithms.

The state-of-the-art results currently consist of an exact $n^{1/3}$-SSP (and thus SSSP) in $\tildeBigO{n^{1/3}}$ rounds due to \cite{censor2020distance}. The same authors also achieved $\tildeBigO{n^{5/17}}$ rounds $(1+\epsilon)$-approximate SSSP \cite{censor2021sparsity}. Very recently, \cite{schneider2023thesis} achieved near optimal $(1+\epsilon)$ SSSP approximation in polylogarithmic time, relying on the minor-aggregation framework established by \cite{rozhovn2022undirected}.

We note that most \hybrid algorithms so far are randomized, with the exception of the algorithms of \cite{anagnostides2021deterministic}, which achieved $(\log{n}/\log{\log{n}})$-approximate APSP in $\tildeBigO{\sqrt{n}}$ rounds, together with a derandomization of \kdis for regimes of $k\geq{}n$, running in $\tildeBigO{\sqrt{k}}$ rounds. 

Further work in the \hybrid model was focused on diameter computation and lower bounds \cite{kuhn2020computing, anagnostides2021deterministic}, computing routing schemes \cite{Kuhn2022Routing, coy2022routing}, and more distance related problems \cite{Feldmann2020Fast, censor2020distance, censor2021sparsity}.

\paragraph{Universal Optimality.} The notion of \emph{universal optimality} in the distributed setting was first offered by Garay, Kutten and Peleg in FOCS '93 \cite{Garay1998}, where they ask, loosely speaking, if it is possible to identify inherent graph parameters that are associated with the distributed complexity of various fundamental network problems, and develop universally optimal algorithms for them.

A line of work in the \congest model that made significant advances towards algorithms for non-worst-case graphs is the low-congestion shortcut framework, introduced by \cite{ghaffari2016distributed}, and further advanced in many subsequent works \cite{Haeupler2021,ghaffari_dfs,ghaffari_mixing_time,haeupler2016near, haeupler2018faster, haeupler2018minor, haeupler2022hop, kitamura2021low, rozhovn2022undirected, zuzic2018towards, Ghaffari2022, zuzic2022universally,ghafarriOPTBroadcast}. The notion of universal optimality is formalized in the work of \cite{haeupler2021universally}, where they explore different notions of universal optimality, and solve many important problems.

\paragraph{Misc.} \hybrid relates to other studied networks of hybrid nature, such as \cite{gmyr2017distributed, Afek1990,baDISC}. Another close model to \hybrid, is the Computing with Cloud (\cwc) introduced by \cite{afek2021distributed}, which consider a network of computational nodes, together with (usually one) passive storage cloud nodes. They explore how to efficiently run a joint computation, utilizing the shared cloud storage and subject to different capacity restrictions. We were inspired by their work on how to analyze neighborhoods of nodes in order to define an optimal parameter for a given graph which limits communication.

\section{Preliminaries}
We consider undirected, connected graphs $G = (V, E, \omega)$, $n = |V|, m = |E|$, with a weight function $\omega$, with weights which are all polynomial in $n$. If the graph is unweighted, $\omega \equiv 1$. The distance between two nodes $v,w\in{V}$ is denoted by $d(v,w)$. The hop-distance between $v,w\in{V}$ is denoted by $\hop(v,w)$ and is the unweighted distance between two nodes. The diameter of a graph is denoted by $D$. Denote by $d^h(v,w)$ the weight of the shortest path between $v$ and $w$ when considering all paths of length at most $h$. 
Let $N_i(v) =\{w \,| \, \hop(v,w)=i\}$ be the set of nodes with hop distance exactly $i$ from $v$, and $B_t(v)=\{w\in{}V \, | \, \hop(v,w)\leq{}t\}=\dot{\bigcup}_{0\leq{}i\leq{}t}N_i(v)$ be the ball of radius $t$ centered at $v$.
Given a set of nodes $C\subseteq V$, the weak diameter of $C$ is defined as $\max_{u,v\in{}C}\hop(u,v)$, where the hop-distance is measured in the original graph $G$. The strong diameter of $C$ is the diameter of the induced subgraph by $C$ in $G$. For any positive integer $x$, let $[x] = \{1, \dots, x\}$.

\paragraph{Model Definitions.}
We formally define \hybrid and \hybridzero.

\begin{definition}[\hybrid model \cite{augustine2020shortest}]
\label{def:hybrid0}
    We consider a network $G=(V,E)$ where $|V|=n$ and the identifiers of the nodes are in the range $[n]$. Communication happens in synchronous rounds. In each round, nodes can perform arbitrary local computations, following which they communicate with each other. Local communication is modeled with the \local model \cite{linial1992locality}, where for any $e = \{v, u\} \in E$, nodes $v$ and $u$ can communicate any number of bits over $e$. Global communication is modeled with the \nccshort model, \cite{augustine2019distributed}, where every node can exchange $O(\log{n})$-bit messages with up to any $\log{n}$ nodes in $G$. It is required that each node is the sender and receiver of at most $\log n$ messages per round.
    
\end{definition}

\begin{definition}[\hybridzero model \cite{augustine2020shortest,augustine2021distributed}]
    The \hybridzero model is like the \hybrid model, with the exception that identifiers are arbitrary and in the range $[n^c]$ for some constant $c$. This implies that a node might not know which identifiers are used in the graph and thus can only send messages to nodes whose identifiers it knows. That is, global communication is over \ncczero \cite{augustine2021distributed} instead of \nccshort. It is assumed that at the start of an algorithm, a node knows its own identifier and the identifiers of its neighbors.
\end{definition}

It is possible to parameterize the \hybrid model by maximum message size $\lambda$ for local edges, and number of bits $\gamma$ each node can exchange per round, in total, via the global edges. As stated, we consider the standard $\lambda=\infty$ and $\gamma=O(\log^2{n})$. The standard distributed models are also specific cases of this parameterization (up to constants), where \local is $\lambda=\infty,\gamma=0$, \congest is $\lambda=O(\log{n}),\gamma=0$, \clique is $\lambda=0,\gamma=n\log{n}$ and \nccshort is $\lambda=0,\gamma=O(\log^2{n})$. 

\paragraph{Communication Primitives.} As a basic tool, we show the following in \hybridzero. To do so, we use  \cite{augustine2019distributed} which show a similar result for a model similar to, but slightly different than, \hybridzero.

\begin{restatable}{lemma}{treeWithoutIDs}
\label{lem:construct_tree_without_ids}
In \hybridzero, it is possible to construct a virtual tree $T$ which spans $G$, has constant maximal degree, and has $\tildeBigO{1}$ depth. It is guaranteed that by the end of the algorithm, every two neighboring nodes in the tree know the identifiers of each other in $G$. This takes $\tildeBigO{1}$ rounds.
\end{restatable}

The proof of \cref{lem:construct_tree_without_ids} is deferred to \cref{appendix:deterministic_ID}. Due to \cref{lem:construct_tree_without_ids}, we get the following solution for \kagg just for the special case of $k=1$. Specifically, note that this result is utilized when showing our \kdis and \kagg algorithms for general $k$.

\begin{lemma}
\label{lem:hybrid_zero_aggregate}
    For $k = 1$, it is possible to solve \kagg  deterministically in $\tildeBigO{1}$ rounds in \hybridzero.
\end{lemma}

\paragraph{Problem Definitions.} We provide formal definitions for problems we solve in \hybrid and \hybridzero which are not already defined above.

\begin{definition}[All-Pairs Shortest Paths (APSP)]
    Every node $v$ must output for every node $w$ the distance $d(v, w)$. In $\alpha$-approximate APSP, $v$ outputs $\hat{d}(v,w)$ for all $w\in{V}$, where $d(v,w)\leq{}\hat{d}(v,w)\leq{}\alpha\cdot{}d(v,w)$. Note that every $v \in V$ must know the identifiers of all nodes in order to be able to write down the output.
\end{definition}

\paragraph{Universal Optimality.}
We follow the approach of \cite{haeupler2021universally}, and define a universally optimal algorithm as follows. Given a problem $P=(S, I)$, split its input into a fixed setting $S$, and parametric input $I$. For example, in \kdis we fix the graph $G$ and the starting locations of all $k$ messages, yet the contents of the messages are arbitrary. For a given algorithm $A$ solving $P$ and any possible state $s$ for $S$ and $i$ for $I$, denote by $t(a, s, i)$ the round complexity of $A$ when run on $P$ with $S=s$ and $I=i$. An algorithm $A$ is universally optimal w.r.t.~$P$ if, for any choice of $s$, the worst case round complexity of $A$ is at most $\tildeBigO{1}$ times that of the best algorithm $A_s$ for solving $P$ which knows $s$ in advance. Formally, for all possible $s$ and any algorithm $A_s$, set $t(A, s) = \max_{i}t(A, s, i)$, and $t(A_s) = \max_{i}t(A_s, s, i)$, it holds that $t(A, s) = \tildeBigO{t(A_s)}$. That is, one must fix a single $A$ that works for all $s$, yet $A_s$ can be different for each $s$.

\paragraph{Miscellaneous.}

\begin{definition}[Square Grid Graph]
    A $d$-dimensional square grid graph $G=(V,E)$ where $n=m^d$, is the cartesian product graph of $d$ $m$-node paths $P_m$. Formally, $G=P_m^1 \times \dots \times P_m^d$. 
\end{definition}

\section{Universally Optimal Broadcast}
\label{sec:optimal_broadcast}

We now show our universally optimal broadcasting result. We begin with some basic properties of $\tk$ in \cref{sec:optimal_broadcast;subsec:basic_properties}, continue with the lower bound in \cref{sec:optimal_broadcast;subsec:lower_bound} and then show the upper bound in \cref{sec:optimal_broadcast;subsec:upper_bound}. Finally, in \cref{sec:optimal_broadcast;subsec:aggregation} we show a bidirectional reduction from \kagg to \kdis in \hybrid, to achieve a universally optimal result for \kagg.

\subsection{Basic Properties of \texorpdfstring{$\tk$}{Tk}}
\label{sec:optimal_broadcast;subsec:basic_properties}

We show the following useful tools regarding $\tk$. Their proofs are deferred to \cref{appendix:basic_properties_of_tk}.

\begin{restatable}{lemma}{tkComputation}
\label{thm:tk_computation}
    It is possible to compute $\tk$ and make it globally known in $\widetilde{O}(\tk)$ rounds in \hybridzero.
\end{restatable}

We also get the following corollary showing that nodes learn more information in the above algorithm, and not just the value of $\tk$.

\begin{corollary}
 When computing $\tk$ using \cref{thm:tk_computation}, every node learns the $\tk(v)$ distribution across the graph, i.e for any $t$, how many nodes $v$ have $\tk(v)=t$.
 \end{corollary}

We now show a statement which limits the rate of growth of $\tk$ as $k$ grows. The statement says that for $k'=\alpha{}k$, the value $\tk[k']$ can only be larger than $\tk$ by a factor which is roughly $\sqrt{\alpha}$. The idea behind the proof is that for any graph, all neighborhoods of radius $\tk$ can learn $k$ messages in $\tk$ rounds. Therefore, if we increase $\tk$ by a factor of $\alpha$, then all neighborhoods of size $\tk \cdot \alpha$ can learn $k \cdot \alpha^2$ messages in $\tk \cdot \alpha$ rounds. For the full proof, see \cref{appendix:basic_properties_of_tk}.

\begin{restatable}{lemma}{growthOfTk}
\label{lem:growth_of_tk}
For $\alpha \geq 1$, 
$\tk[\alpha{}k]\leq{}6\sqrt{\alpha}\cdot{}\tk$.
\end{restatable}

We now bound $\tk$ in terms of $n,D,k$, with the proof deferred to \cref{appendix:basic_properties_of_tk}.

\boundingTk*

\subsection{Lower Bound}
\label{sec:optimal_broadcast;subsec:lower_bound}

We now desire to prove \cref{thm:kdisLB}.

\kdisLB*

Before doing so, we show the following lemma which basically states that we can assume messages are sent \emph{as-is} over the global network, without any coding techniques to compress the amount of bits which need to be transferred. Note that the following statement is applicable to showing the universal lower bound for \kdis in \hybrid, as the universal lower bound is w.r.t.~the graph $G$ and initial locations of the messages, but not their contents. In essence, the following lemma states that if the messages are uniformly chosen random strings, then it is not possible (even for a randomized algorithm) to compress the messages by more than a constant factor.

\begin{lemma}
\label{lem:no_coding_tricks}
Let there be a graph $G$ and a partition of the nodes $V = A \cup B$, $A\cap B = \emptyset, A \neq \emptyset, B \neq \emptyset$, and some value $k$. Assume each node in $A$ is given the freedom to choose some number arbitrary messages, such that all nodes in $A$ in total choose $k$ arbitrary $b = O(\log n)$-bit messages. Then, at least $\Omega(kb)$ bits of information must be communicated from $A$ to $B$ in order for the nodes in $B$ to be able to reconstruct all $k$ messages with success probability at least $1/2$. This holds even if every nodes knows the entire topology of $G$ and how many messages each node in $A$ gets to choose.
\end{lemma}
\begin{proof}
    Let $ALG$ be the optimal algorithm (in terms of minimal bits sent between $A$ and $B$) which performs communication between the nodes in $A$ and those in $B$ such that at the end of its execution, each of the $k$ messages is recovered by at least one node in $B$, no matter what the contents of the messages are; it is assumed that $ALG$ succeeds (i.e.,~all messages are recovered) with probability at least $1/2$. Denote by $r$ the number of bits which $ALG$ transfers from $A$ to $B$ in the worst case. We now desire to show that $r = \Omega(kb)$.

    Each node in $A$ that can choose messages simply chooses random bits, uniformly at random, for each of its messages. In total, the nodes in $A$ choose $kb$ random bits, uniformly at random.

    Observe that there are potentially $2^r$ bit strings which $ALG$ can send from $A$ to $B$ -- denote the set of these strings by $S$. When receiving a string $s \in S$, the nodes in $B$ perform some algorithm at the end of which they state that the $k$ messages chosen in $A$ are some strings $m = m_1, \dots, m_k$. For each $s \in S$, the nodes in $B$ have some probability distribution over the $m = m_1, \dots, m_k$ messages which they believe $A$ has, denote this distribution by $m_s = P(m | s)$.

    For any selection of messages $m = m_1, \dots, m_k$ chosen by $A$, denote by $s_m = P(s | m)$ the distribution over strings in $S$ which $ALG$ sends to $B$ given $m$. Clearly, as $ALG$ always succeeds with probability at least half, $1/2$ for any specific $m'$ it holds that $\sum_{s \in S} s_{m'}\cdot m'_s \geq 1/2$. Denote by $M$ the set of all messages that the nodes in $A$ can choose -- i.e.,~the set of all strings of length $kb$ bits. Summing over all possible $m' \in M$, we get that 

    \[
        \sum_{m' \in M} \sum_{s \in S} s_{m'}\cdot m'_s \geq |M|/2 = 2^{kb-1}.
    \]

    Notice that $s_{m'} \leq 1$ always, corresponds to a probability of an event, and for any given $s \in S$, it holds that $\sum_{m' \in M} m'_s = 1$, as it corresponds to a sum of probabilities of disjoint events which together partition the event space. Plugging both of these into the above gives

    \[
        2^{kb-1} \leq \sum_{m' \in M} \sum_{s \in S} s_{m'}\cdot m'_s =  \sum_{s \in S} \sum_{m' \in M} s_{m'}\cdot m'_s \leq \sum_{s \in S} \sum_{m' \in M} m'_s = \sum_{s \in S} 1 = |S| = 2^r.
    \]

    Thus, $r \geq kb-1$ and so $r = \Omega(kb)$, as required.
\end{proof}

Now, using \cref{lem:no_coding_tricks}, we can show \cref{thm:kdisLB}.

\begin{proof}
    We denote by $A_M$ the optimal amount of rounds to broadcast all messages $M$ to all of $G$. We desire to show that $A_M = \tilde \Omega(T_k)$. Further, for a set of nodes $V' \subseteq V$, denote by $M(V')$ the messages in $M$ which are originally stored at any node in $V'$.
    
    Let $v^*=\argmax_{v\in{}V}\tk(v)$.  Throughout the proof, we use the following observation: if a message is at distance $\ell+1$ from $v^*$, then in order for $v^*$ to receive it in $\ell$ rounds, it has to be sent at least once into $B_{\ell}(v^*)$ through the global network. Due to \cref{lem:no_coding_tricks}, we know that one cannot compress the information to be sent by more than at most a constant factor, so it is possible to assume that messages are just sent \emph{as-is}, without any coding techniques to shorten specific messages or sets of messages.

    We now split to cases.
    
    \paragraph{Case 1: $\tk<D/2$ and $|M(V \setminus B_{\tk}(v^*))| > k/2$.}
    We bound the global network bandwidth capacity of $B_{\tk-1}(v^*)$. Since $\tk$ is the minimal radius s.t.~$|B_{\tk}(v^*)|\geq{}k/\tk$, then $|B_{\tk-1}(v^*)|<k/\tk$. Each node can receive $\log{n}$ messages per round using the global network, so in $(\tk-1)/(2\log{n})$ rounds, $B_{\tk-1}(v^*)$ can at most receive 
    \[
        \frac{\tk-1}{2\log{n}}\log{n}|B_{\tk-1}(v^*)|
        \leq{}
        \frac{\tk-1}{2}\frac{k}{\tk}
        \leq{}
        \frac{k}{2}
    \]
     messages using the global network. 
    
    As $|M(V \setminus B_{\tk}(v^*))| > k/2$, then also $P = M(V \setminus B_{\tk-1}(v^*)), |P| > k/2$. In order for $v^*$ to receive the set of messages $P$ in at most $\tk-1$ rounds, they have to be sent through the global network into $B_{\tk-1}(v^*)$. However, as we just showed, it takes at least $(\tk-1)/(2\log{n})$ rounds for $B_{\tk-1}(v^*)$ to receive $k/2$ messages using the global network. Therefore, $A_M>\min\{(\tk-1)/(2\log{n}), \tk-1\} = (\tk-1)/(2\log{n})$.

    \paragraph{Case 2: $\tk<D/2$ and $|M(V \setminus B_{\tk}(v^*))| \leq k/2$.}
    Now, there are at least $k/2$ tokens inside $B_{\tk}(v^*)$, and we split to cases again. 
    We look at the two \emph{rings} surrounding $v^*$, denoted $R_1 = B_{\tk/2}(v^*)$ and $R_2 = B_{\tk}(v^*)\setminus{}R_1$. 

    \paragraph{Case 2.1: $|M(R_2)| \geq k/4$.}
    If the outer ring $R_2$ has at least $k/4$ tokens, then we proceed similarly to Case 1 above. The global capacity of $R_1$ in $\tk/(4\log{n})$ rounds is at most $k/4$, so $A_M>\tk/(4\log{n})$.

    \paragraph{Case 2.2: $|M(R_2)| < k/4$.}
    As $R_2$ has less than $k/4$ tokens, then $R_1$ has at least $k/4$ tokens. We now flip our point of view from receiving capacity to transmitting capacity: in order for a node $w\notin{}B_{\tk}(v^*)$ to receive all tokens in less than $\tk/4$ rounds, all the tokens in $R_1$ have to be sent through the global network. The of transmitting capacity of $R_1$ in $\tk/(4\log{n})$ rounds over the global network is the same as the receiving capacity, which is bounded by $k/4$ by the same arguments as Case 2.1. Therefore at least $\tk/(4\log{n})$ rounds are required to send all the tokens from $R_1$ to $w$, which means $A_M>\tk/(4\log{n})$.
    Note that this holds only if there exist a node $w\notin{}B_{\tk}(v^*)$. If not, then $B_{\tk}(v^*)=V$, but this would mean that $\tk\geq{}D/2$, contradicting the assumption throughout Case 2.

    \paragraph{Case 3: $\tk\geq{}D/2$.}
    Now assume $\tk\geq{}D/2$. It means that $|B_{D/2-1}(v^*)|<k/\tk$ by \cref{def:bcast_quality}, and there exists a node $w\notin{}B_{D/2-1}(v^*)$, or otherwise the diameter would be smaller.
    We can repeat the same arguments from Cases 1 and 2 above, with $D/2-1$ instead of $\tk$, and get the same result.

    In all cases, we showed $A_M = \tilde \Omega(\tk)$, and so we are done. 
\end{proof}

\subsection{Upper Bound}
\label{sec:optimal_broadcast;subsec:upper_bound}

We now prove \cref{thm:optimal_dissemination}.

\kdisUP*

We recall a known result about computing $(\alpha,\beta)$-ruling sets in \congest.

\begin{definition}
\label{def:ruling_sets}
    An $(\alpha,\beta)$-ruling set for $G=(V,E)$ is a subset $W\subseteq{}V$, such that for every $v\in{}V$ there is a $w\in{}W$ with $\hop(v,w)\leq{}\beta$ and for any $w_1,w_2\in{}W$, $w_1\neq{}w_2$, we have $\hop(w_1,w_2)\geq{}\alpha$.
\end{definition}

\begin{theorem} [Theorem 1.1 in \cite{kuhn2018deterministic}]
\label{thm:ruling_sets}
    Let $\mu$ be a positive integer. A $(\mu+1, \mu\ceil{\log{n}})$-ruling set can be computed deterministically in the local network in $O(\mu\log{n})$ rounds in \congest.
\end{theorem}

We use the following terms throughout the proof, which we define formally.

\begin{definition} [Flooding]
\label{def:flooding}
    Flooding information through the local network, is sending that information through all incident local edges of all nodes. On subsequent rounds, the nodes aggregate the information they received and continue to send it as well. After $t$ rounds, every node $v$ knows all of the information which was held by any node in its $t$-neighborhood before the flooding began.
\end{definition}

\begin{lemma}[Uniform Load Balancing]
\label{lemma:load_balancing}
    Given a set of nodes $C$ with weak diameter $d$ and a set of messages $M$ with $|M|=k$ distributed across $C$, there is an algorithm that when it terminates, each $v\in{}C$ holds at most $\ceil{k/|C|}$ messages. The algorithm runs in $2d=O(d)$ rounds. We say that $C$ uniformly distributes $M$ within itself.  
\end{lemma}
\begin{proof}
    In $d$ rounds, all nodes flood the messages and identifiers of $C$. The minimal identifier node then computes an allocation such that each $v\in{}C$ is responsible for at most $\ceil{k/|C|}$ messages, and floods the allocation for another $d$ rounds, so it reaches all $v\in{}C$.
\end{proof}

We use \cref{thm:ruling_sets} to prove the following lemma on creating clusters with low weak diameter and roughly the same number of nodes.

\begin{lemma}
\label{lem:clustering_with_weak_diam}
    For any $k$, it is possible to partition the set of nodes into clusters with weak diameter at most $4\tk\ceil{\log{n}}$ such that each cluster has between $k/\tk$ and $2k/\tk$ nodes. This lemma returns a set $R \subseteq V$ of cluster leaders, for every $r \in R$, denote by $C(r)$ the cluster which $r$ leads, and for every cluster $C$ denote its leader by $r(C) \in R$. Every node knows whether or not it is in $R$ and also knows to which cluster it belongs. This all takes $\tildeBigO{\tk}$ rounds in \hybridzero.
\end{lemma}
\begin{proof}
    We compute $\tk$ in $\tildeBigO{\tk}$ rounds by \cref{thm:tk_computation}. We choose $\mu=2\tk$ and use \cref{thm:ruling_sets} to compute a $(2\tk+1, 2\tk\ceil{\log{n}})$-ruling set in $O(\tk\log{n})$ rounds\footnote{\cref{thm:ruling_sets} is stated in \congest. Clearly, it can be run in \hybrid. It is potentially an intersting question if it can be executed in \hybridzero, as \congest might assume that the identifiers in the graph are from some specific pallet, e.g.,~$[n]$. To overcome this assumption, we execute \cref{lem:construct_tree_without_ids} to construct a virtual tree over all the nodes and use it within $\tildeBigO{1}$ rounds to rename the nodes to have identifiers in whatever set the algorithm in \congest assumes. The nodes assume these new identifiers just for the execution of \cref{thm:ruling_sets}, and then return to use their original identifiers.}. We denote the set of rulers by $R$. Then, for $2\tk\ceil{\log{n}}$ rounds, each node learns its $2\tk\ceil{\log{n}}$ neighborhood and the ruling nodes in it, through the local network. For every $v\in{V}$, let $r(v)$ be the closest ruling node by hop distance, with ties broken by minimum identifier. By \cref{def:ruling_sets}, $r(v)$ must be in its $2\tk\ceil{\log{n}}$ neighborhood. By exploring this neighborhood, each node $v$ finds $r(v)$.

For any $r \in R$, define the cluster of $r$ as $C(r)$. For any cluster $C'$, let $r(C')$ be the $r \in R$ such that $C(r(C')) = C'$. Every node $v \in V\setminus R$ joins the cluster of its closest ruling node $r(v)$. Notice that any cluster $C$ contains exactly one ruling node, $r(C)$, and set the cluster identifier of $C$ as the identifier of $r(C)$. \cref{def:ruling_sets} guarantees that the weak diameter of each such cluster is at most $2\beta=4\tk\ceil{\log{n}}$. Thus, for $4\tk\ceil{\log{n}}$ rounds, each node $v$ floods $r(v)$ through the local network, so for every cluster $C$, any $v\in{}C$ knows all the nodes in $C$. 

Let $C$ be a cluster. As for every $r_1\neq{}r_2\in{}R$, $\hop(r_1,r_2)\geq{}\alpha=2\tk+1$, it holds that $B_{\tk}(r(C))\subseteq{}C$ -- that is, every node in $u \in B_{\tk}(r(C))$ joins $C$, as the closest ruling node to $u$ is $r(C)$. By \cref{def:bcast_quality}, $|C|\geq{}|B_{\tk}(r(C))|\geq{}k/\tk$. Thus, every cluster has minimum size $k/\tk$.

Now, we make sure that our clusters are not too big. Each cluster $C$ with $|C|>2k/\tk$ splits deterministically to more clusters, until each cluster holds $k/\tk\leq{}|C|\leq{}2k/\tk$. This can be computed locally for each cluster, for example by greedily assigning groups of $2k/\tk$ node identifiers inside the cluster to the new cluster, and choosing the leader as the minimal identifier node. After this process, we get at most $n\tk/k$ disjoint clusters, each with weak diameter at most $4\tk\ceil{\log{n}}$, and of size $k/\tk\leq{}|C|\leq{}2k/\tk$. 
We add every leader of the new clusters which were split to the set $R$.

\end{proof}

Finally, we also show the following helper lemma on pruning trees.

\begin{lemma}
\label{thm:ncc_tree_pruning}
    Let there be a tree $T=(V,E_T)$ with root $r$, constant maximal degree and depth $d$. Given some function $f: V \rightarrow \{0, 1\}$, there is an algorithm that constructs a tree $T'=(U,E_{T'})$, $U=\{v \mid v \in V, f(v) = 1\}\subseteq{}V$, with constant maximal degree and depth $d'\leq{}d$. This takes $O(d^2)$ rounds in \hybridzero. It is assumed that for every $v\in V$, the value $f(v)$ is known to $v$ before this algorithm is run.
\end{lemma}
\begin{proof}
    Denote $U = \{v \mid v \in V, f(v) = 1\}$. Notice that every $v\in V$ knows if $v\in U$. For every $v \in V$, denote by $T(v)$ the subtree of $T$ rooted at $v$. Now, every $v \in V$ computes $|U \cap T(v)|$. This is done by each node sending up the tree how many nodes in $U$ are in its subtree. This takes $O(d)$ rounds.

    Now, the root node $r$ observes itself. If $|U \cap T(r)| = 0$, the algorithm halts and we return an empty tree. Otherwise, if $r \in U$, then it does nothing. If $r \notin U$, it finds some arbitrary node $v \in U$, and \emph{swaps} positions with $v$ -- both $r$ and $v$ inform their parents and children in $T$ that they swap positions, i.e.,~$v$ is now the root of the tree and $r$ now occupies the position which $v$ previously did. In either case, the tree is now rooted by a node from $U$, and we recurse on the subtrees of the children of the root.

    Notice that to find a node $v \in U$, node $r$ simply performs a walk down the tree, each time choosing to go to a child that has some nodes of $U$ in its subtree. Further, notice that we can perform all the recursive steps in parallel, as we recurse on disjoint subtrees. Finally, whenever a subtree contains only nodes from $V \setminus U$, then the entire subtree is removed as the root of that subtree will halt the recursion and return an empty tree. Thus, the tree we are left with in the end is $T'$, it has only nodes from $U$ and depth $d' \leq d$.

    All in all, each step of the recursion takes $O(d)$ rounds, and we have $O(d)$ recursive steps, resulting in an $O(d^2)$ round complexity.
\end{proof}

\begin{figure}
  \centering
\includegraphics[width=0.4\textwidth]{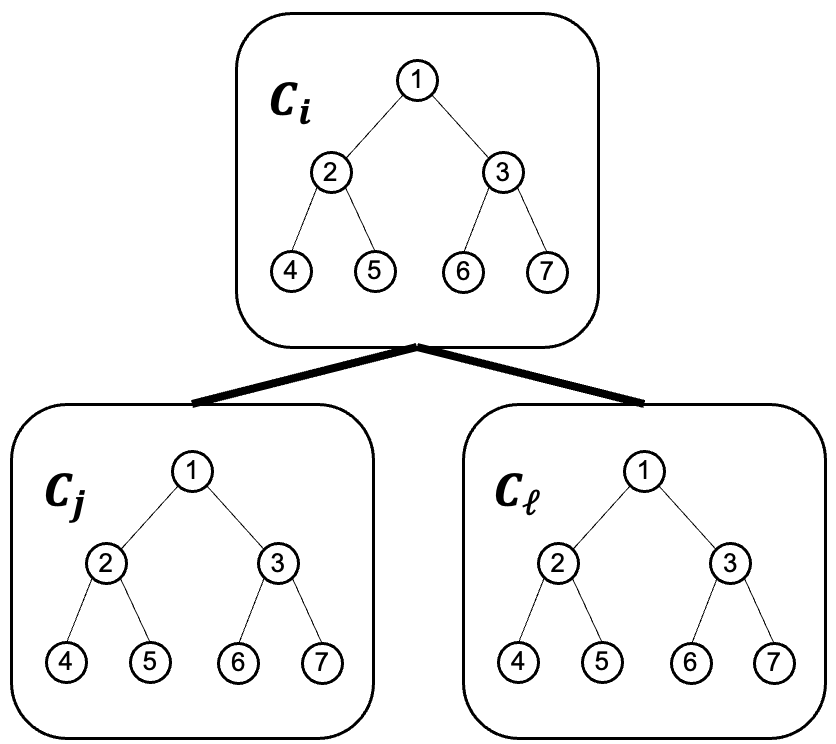}
  \caption{Overview of the proof of \cref{thm:optimal_dissemination}. We first create clusters with $\tildeBigO{\tk}$ weak diameter and roughly the same number of nodes. Then, we construct a logical tree of the clusters, with constant maximal degree and polylogarithmic depth. Inside each cluster, we create a logical binary tree over the nodes of the cluster. As the clusters are roughly of the same size, we can ensure that the trees inside the clusters have the exact same shape. We ensure that for any two neighboring clusters in the cluster tree, the nodes in their internal trees know their respective nodes in the tree of the other cluster -- i.e.,~node $3$ in cluster $C_i$ knows the identifier of nodes $3$ in $C_j, C_\ell$ and can communicate directly with them using the global edges. Once we are done constructing all of these trees, we propagate all the $k$ messages in the graph up to the top of the cluster tree, and then propagate them back down to all the clusters, to ensure that every node in the graph receives all the messages.} 
  \label{fig:notations}
\end{figure}

We are now ready to prove \cref{thm:optimal_dissemination}.

\begin{proof}
\label{optimal_dissem_proof}
The algorithm consists of several phases: clustering, cluster-chaining, load balancing and dissemination (see \cref{fig:notations}). The clustering phase ensures that we partition the nodes to disjoint clusters of similar size, such that the weak diameter of each cluster is small. 
In the cluster-chaining phase, we order the clusters in a logical tree with constant degree and polylogarithmic depth, and let the nodes of each cluster know the nodes of its parent and children clusters. 
In the dissemination phase, we trickle all the tokens up to the root cluster using the global network, and the chaining we devised in the cluster-chaining phase. Then, we trickle the tokens down the tree, such that each cluster learns all the tokens. 

We begin by computing $k$ by summing how many tokens each node holds, using \cref{lem:hybrid_zero_aggregate}, in $\tildeBigO{1}$ rounds. Then, we compute $\tk$ in $\tildeBigO{\tk}$ rounds by \cref{thm:tk_computation}.

\paragraph{Clustering.} We wish to create a partition of the nodes in the graph into clusters of roughly the same size and with small weak diameter. We execute \cref{lem:clustering_with_weak_diam} with $k$ and receive $R$, the set of cluster leaders. This takes $\tildeBigO{\tk}$ rounds.

\paragraph{Cluster-Chaining.} This phase consists of two sub-phases. We first create a logical tree of the clusters, denoted $\mathbb{T}_C$, with constant maximal degree and depth at most $\tildeBigO{1}$.
Then, within each cluster $C_i$, we order its nodes in a logical binary tree $\mathbb{T}_i$. Finally, we use the internal trees to associate nodes of one cluster with the nodes of its parent and children clusters.

\paragraph{Building the cluster tree.}
We run \cref{lem:construct_tree_without_ids} in $\tildeBigO{1}$ rounds to obtain a virtual tree which spans $G$ of constant maximal degree and depth $\tildeBigO{1}$.
After the clustering phase, each node knows whether it is a cluster leader or not. Thus, we define a function $f$ where every $v \in V$ sets $f(v) = 1$ if $v \in R$, and $f(v)=0$ otherwise. We now use \cref{thm:ncc_tree_pruning} with $T, f$ to compute a tree $\mathbb{T}_C$, with constant maximal degree and depth $\tildeBigO{1}$, of cluster leaders. This takes $\tildeBigO{1}$ rounds.

\paragraph{Matching parent and children cluster nodes.} Observe a cluster $C_i$. Recall that every node in $C_i$ knows all of the other nodes in $C_i$, and so they each computes a logical binary tree $\mathbb{T}_i$ of the nodes in $C_i$, with $r(C)$ as the root of $\mathbb{T}_i$. We desire for $\mathbb{T}_i$ to have exactly $2k/\tk$ nodes. As $k/\tk\leq{}|C_i|\leq{}2k/\tk$, then we just append more nodes from $C_i$ to $\mathbb{T}_i$, potentially repeating every node in $C_i$ twice in $\mathbb{T}_i$.

Let $C_i, C_j$ be two clusters whose leaders $r(C_i), r(C_j)$ are neighbors in the cluster tree $\mathbb{T}_C$ -- w.l.o.g., assume $r(C_i)$ is the parent of $r(C_j)$ in $\mathbb{T}_C$. It holds that $\mathbb{T}_i, \mathbb{T}_j$ have the same structure, as all these internal trees have the same number of nodes and are constructed virtually to have the same structure. Let $v_i \in \mathbb{T}_i, v_j \in \mathbb{T}_j$ be two nodes with the same position in their trees (same level of the tree, same index within the level). We now desire for $v_i$ and $v_j$ to be made aware of each other -- that is, to learn the identifiers of each other so that they can communicate over the global network.

We begin with $r(C_i), r(C_j)$, who are at the root of $\mathbb{T}_i, \mathbb{T}_j$, respectively. They already know the identifiers of each other, as that is guaranteed by the construction of $\mathbb{T}_C$. Let $L_i, R_i$ be the children of $r(C_i)$, and $L_j, R_j$, those of $r(C_j)$. Node $r(C_i)$ sends to $r(C_j)$ the identifiers $L_i, R_i$, and $r(C_j)$ sends $L_j, R_j$ to $r(C_i)$. Now, $r(C_i)$ sends to $L_i, R_i$ the identifiers $L_j, R_j$, and likewise $r(C_j)$ communicates with $L_j, R_j$. All of this takes $O(1)$ rounds using the global network, and can be done in parallel for any $r(C_i), r(C_j)$ which are neighbors in $\mathbb{T}_C$. 

Notice that now $L_i, L_j$ know both their identifiers, and likewise $R_i, R_j$. Thus, they each continue down their respective subtrees. As the trees have $O(\log n)$ depth, and each level of the trees takes $O(1)$ rounds to process, this takes $O(\log n)$ rounds in total.

\paragraph{Load balancing.} Each cluster $C_i$ uniformly distributes the tokens of its nodes within itself by \cref{lemma:load_balancing}. There are $k$ tokens in the graph, and so at most $k$ tokens in $C_i$. Further $|C_i|\geq{}k/\tk$, so $C_i$ can load balance its tokens such that each node has at most $\tk$ tokens. In total, this phase takes $\tildeBigO{\tk}$ rounds, because the weak diameter of $C_i$ is at most $4\log{n}\tk$.

\paragraph{Dissemination.} We now aim to gather all the tokens in the root cluster $C_r$ of the cluster tree $\mathbb{T}_C$. 
For $\tildeBigO{1}$ iterations, we send the tokens of each cluster up the cluster tree. Each node holds at most $\tk$ tokens and is matched with at most $2$ nodes in the parent cluster, 
so in $2\tk$ rounds we can send the tokens using the global communication network. In the beginning of each iteration, each cluster again load balances the tokens it received in the last iteration, by \cref{lemma:load_balancing}. This is done to prevent the case of a node holding more than $\tk$ tokens, because it could receive up to $\tk$ tokens from each child cluster. We then continue to send the tokens up to the root, which takes at most $\log{n}$ iterations by the depth of the cluster tree $\mathbb{T}_C$. Considering the load balancing at the beginning of each iteration, this step takes $\log{n} \cdot (\tk+4\tk\ceil{\log{n}})=\tildeBigO{\tk}$ rounds.

Now, the root cluster holds all of the tokens. It again load balances the tokens within its nodes, such that each node holds at most $\tk$ tokens. We now send down the tokens in the same manner, down the cluster tree. For $\tildeBigO{1}$ iterations, each node sends its at most $\tk$ tokens to its matched nodes in the at most $2$ children clusters, through the global network. Again we load balance at every iteration, to prevent accumulation of more than $\tk$ tokens in each node. This is necessary because the matching can match 2 nodes in $C_i$ to one node in $C_j$. Now each cluster holds all the tokens. Each node floods all tokens through the local network for $4\tk\ceil{\log{n}}$, the weak diameter of a cluster, making all nodes in its cluster learn all the tokens. This phase takes  $\log{n} \cdot(\tk+4\tk\ceil{\log{n}})+4\tk\ceil{\log{n}}=\tildeBigO{\tk}$ rounds.

We now solved \kdis, since every node in $G$ knows all of the tokens. Summing over all the phases, the algorithm takes $\tildeBigO{\tk}$ rounds.
\end{proof}

\subsection{Universally Optimal Aggregation}
\label{sec:optimal_broadcast;subsec:aggregation}

We now show a universally optimal solution for the \kagg problem, using a bidirectional reduction from \kdis.

\kaggLBUB*

\begin{proof}
    We show a bidirectional reduction from \kdis. Given a graph $G=(V,E)$, if there is an algorithm that solves \kagg in $t$ rounds, we can solve \kdis in $\tildeBigO{t}$ rounds, and vice versa. 

    First, if there is an algorithm solving \kagg in $t$ rounds, we can employ it to solve \kdis in $\tildeBigO{t}$ rounds. Intuitively, since we have $k$ tokens to disseminate and $k$ aggregation results that can be made globally known in $t$ rounds, we would like to place those $k$ tokens in different indices of the values, and have the rest of the nodes send the unit element in the rest of the indices. The only problem is that all of the nodes holding tokens need to coordinate in which indices each node should put its tokens, so they match the $k$ tokens to $k$ indices. This can be done by the following algorithm, operating in $\tildeBigO{1}$ rounds.

    First, we use \cref{lem:construct_tree_without_ids,thm:ncc_tree_pruning} to construct a tree of all nodes with at least one token to disseminate. After we get this tree $T$, we can in $\tildeBigO{1}$ rounds compute for each node $v$ how many tokens its subtree, including itself, holds. We denote it by $\ell(v)$. This is done by sending the information up from the last level of the tree, aggregating the number in each node, and sending it to the parent node. Then we begin allocating the indices, starting from the root.

    If the root holds $m$ tokens, it reserves the first $m$ indices for itself, and tells its first son it should start allocating from $m$, and to its second son, if exists, it should start allocating from $m+\ell$(first son). The root does continues in this fashion for all its children. The nodes lower in the tree continue in the same fashion. This creates a bijection of the $k$ tokens across the graph to the $k$ indices of aggregation, and correctly allocates all indices to each token-holding node. Then we can run the \kagg algorithm in $t$ rounds, and all nodes learn the $k$ tokens.
    In total, the algorithm takes $t+\tildeBigO{1}=\tildeBigO{t}$ rounds.
    This shows that $\tk$ is a universal lower bound for the \kagg problem.

    Conversely, we show that we can indeed solve \kagg in $\tildeBigO{\tk}$ rounds, deterministically. We note that once \emph{only one} node learns the results of all $k$ aggregate functions, we can disseminate it in $\tildeBigO{\tk}$ rounds by \cref{thm:optimal_dissemination}. We use similar steps to the proof of \cref{thm:optimal_dissemination}. First, we cluster the nodes using the same procedure. We then compute inside each disjoint cluster $k$ intermediate aggregations, and load balance it inside the cluster with \cref{lemma:load_balancing}. That way, each node holds at most $\tk$ aggregation results. Then we use the cluster tree and cluster chaining in the proof to send the intermediate aggregation results up the cluster tree to the root cluster. In each step, we load balance again. As the depth of the constructed cluster tree is at most $\tildeBigO{1}$, this process finishes in $\tildeBigO{\tk}$ rounds. Once all the information is stored in the root cluster, we flood it inside it and compute locally the final $k$ aggregation results. This step takes $\tildeBigO{\tk}$ rounds by the weak diameter of each cluster, which is at most $4\tk\log{n}$.

    Finally, we disseminate the $k$ aggregation results from some node in the root cluster to the entire graph, using \cref{thm:optimal_dissemination} in $\tildeBigO{\tk}$ rounds.
\end{proof}

\section{Applications}
\label{sec:applications}

\subsection{\texorpdfstring{$(1+\epsilon)$}{(1+eps)}-approximate APSP in Unweighted Graphs} 

We now prove \cref{thm:unweighted_apsp_approx}.

\unweightedAPSPApprox*

To do so, we require the novel $(1+\epsilon)$-approximate SSSP result from \cite{schneider2023thesis}. 

\begin{theorem}[Theorem 3.29 from \cite{schneider2023thesis}]
\label{thm:polylog_approx_sssp}
    A $(1+\epsilon)$-approximation of SSSP can be computed in $\tildeBigO{1/\epsilon^2}$ \whp in the \hybrid model.
\end{theorem}

We note that the algorithm of \cite{schneider2023thesis} relies heavily on the novel algorithm presented in \cite{rozhovn2022undirected}.

We now proceed to proving \cref{thm:unweighted_apsp_approx}. See \cref{alg:unweighted_apsp} for an overview of our algorithm.

\begin{proof}
    We begin by computing $\tk[n]$ in $\tildeBigO{\tk[n]}$ rounds using \cref{thm:tk_computation}, so that from now on we can assume all nodes know this value.
    We proceed by broadcasting the identifiers of all the nodes in $\tildeBigO{\tn}$ rounds using \cref{thm:optimal_dissemination}. From now on, we can assume we are in \hybrid instead of \hybridzero, and thus we are able to execute algorithms such as \cref{thm:polylog_approx_sssp}.
    We execute \cref{lem:clustering_with_weak_diam} with $k = n$, in $\tildeBigO{\tk[n]}$ rounds, to cluster the nodes such that we know a set $R$ of cluster leaders, each cluster $C$ has weak diameter at most $4\tn\ceil{\log{n}}$, and $n/\tk[n] \leq |C| \leq 2n/\tk[n]$.

    Now, observe that as the clusters are disjoint and each has size at least $n/\tn$, then we have at most $\tn$ clusters and as such $|R| \leq \tn$. Using \cref{thm:polylog_approx_sssp}, it is possible to compute $(1+\epsilon)$-approximate distances from all the nodes in $R$ to all the graph in $\tildeBigO{|R|/\epsilon^2} = \tildeBigO{\tn/\epsilon^2}$ rounds \whp. Denote the computed approximate distances by $\hat{d}$.

    Each node $v$ learns its $x = (4\tn\ceil{\log{n}})/\epsilon$ neighborhood, denoted $B_{x}(v)$. This takes $O(x) = O(\tn/\epsilon) = O(\tn/\epsilon^2)$ rounds. Then, every node $v$ broadcasts its closest node in $R$, denoted $c_v \in R$, and the unweighted distance $d(v, c_v)$. As each node broadcasts $O(1)$ messages, then using \cref{thm:optimal_dissemination}, this requires $\tildeBigO{\tn}$ rounds.

    Finally, each node $v$ approximates its distance to each node $w$ as follows. If $w \in B_x(v)$, then $v$ knows its exact distance to $w$, as the graph is unweighted, and thus sets $\delta(v, w) = d(v, w)$. Otherwise, $v$ sets $\delta(v, w) = \hat{d}(v, c_w) + d(w, c_w)$, where $c_w$ is the closest node in $R$ to $w$. Note that $v$ knows both $c_w$ and $d(w, c_w)$, as $w$ broadcasts these values in the previous step.

    We conclude the proof by showing that $\delta$ is a $(1+\epsilon)$ approximation of $d$.

    If $w\in{}B_{x}(v)$, $\delta(v,w)=d(v,w)$. Otherwise, $d(v,w)>x=4\tn\ceil{\log{n}}/\epsilon$, and $\delta = \hat{d}(v, c_w) + d(w, c_w)$. We begin by showing that $\delta(v, w) \geq d(v, w)$. As $\hat{d}(v,c_w)$ is a valid $(1+\epsilon)$-approximation, $\hat{d}(v,c_w)\geq{}d(v,c_w)$, and thus $\delta(v, w) \geq d(v, c_w) + d(w, c_w) \geq d(v, w)$, where the last inequality is due to the triangle inequality. We now bound $\delta(v,w)$ from above. It holds that $d(v,w)>4\tn\ceil{\log{n}}/\epsilon$ and $d(w,c_w)\leq{}4\tn\ceil{\log{n}}$, as the weak diameter of each cluster is at most $4\tn\ceil{\log{n}}$. Therefore, $d(w,c_w)\leq{}4\tn\ceil{\log{n}}<\epsilon\cdot{}d(v,w)$. As such, the following holds.
    \begin{align*}
        \delta(v,w)&=\hat{d}(v,c_w)+d(c_w,w)\leq{}(1+\epsilon)\cdot d(v,c_w)+d(w,c_w) \\
        &\leq{}(1+\epsilon)\cdot (d(v,w)+d(w,c_w))+d(w,c_w) \\
    &=(1+\epsilon)\cdot d(v,w)+(2+\epsilon)\cdot d(w,c_w) \\
    &<(1+\epsilon)\cdot d(v,w)+(2\epsilon+\epsilon^2)\cdot d(v,w) \\
    &=(1+3\epsilon+\epsilon^2)\cdot d(v,w)
    \overset{\epsilon'=3\epsilon+\epsilon^2}{=}(1+\epsilon')\cdot d(v,w)
    \end{align*}

    Note that we achieve a $(1+\epsilon')$ approximation, where $\epsilon' = 3\epsilon+\epsilon^2$. As $\epsilon \in (0, 1)$, then $\epsilon' < 4\epsilon$, and so it is possible to run all the above with $\epsilon/4$ to achieve the desired result.
\end{proof}

\begin{algorithm}
\caption{$(1+\epsilon)$-Approximate Unweighted APSP} \label{alg:unweighted_apsp}
        Compute $\tn$

        Broadcast the identifiers of all the nodes.
        
        Run \cref{lem:clustering_with_weak_diam} with $k = n$ to get the set of cluster leaders $R$.
        
        Run $(1+\epsilon)$-approx SSSP from each node in $R$. Denote computed distances by $\hat{d}$.
        
        Learn $x = (4\tn\ceil{\log{n}})/\epsilon$ neighborhood.

        Using \cref{thm:optimal_dissemination}, each node $v$ broadcasts its closest cluster leader $c_v \in R$ and $d(v,c_v)$.
        
        Node $v$ approximates its distance to any $w\in{V}$ by: 
        \[
        \delta(v,w)=
            \begin{cases}
                d(v,w), & w\in{}B_{x}(v)\\
                \hat{d}(v, c_w)+d(c_w,w), & \text{otherwise}
            \end{cases}
        \]
\end{algorithm}

\subsection{APSP Approximations in Weighted Graphs}

We show several algorithms for approximating APSP in weighted graphs. We begin with \cref{thm:firstWeightedAPSPApprox}.

\firstWeightedAPSPApprox*

We employ the well-known technique of computing a spanner -- a subgraph with fewer edges which maintains a good approximation of distances in the original graph.

\begin{theorem} [Corollary 3.16 in \cite{rozhovn2020polylogarithmic}]
\label{thm:spanner_construction}
    Let $G=(V,E,\omega)$ be a weighted graph. There exists a deterministic algorithm in \congest which computes a $(2k-1)$-stretch spanner of size $\widetilde{O}(kn^{1+1/k}\log n)$ in $\widetilde{O}(1)$ rounds.
\end{theorem}

To prove \cref{thm:firstWeightedAPSPApprox}, we execute \cref{thm:spanner_construction} and then broadcast the resulting spanner.

\begin{proof}
    We begin by broadcasting the identifiers of all the nodes in $\tildeBigO{\tn}$ rounds using \cref{thm:optimal_dissemination}. From now on, we can assume we are in \hybrid instead of \hybridzero, and thus execute algorithms such as \cref{thm:spanner_construction}.
    We execute \cref{thm:spanner_construction} with $k=\epsilon\log{n}/2$ and receive a spanner with $O(k\cdot{}n^{1+1/k}\cdot \log n)=\widetilde{O}(4^{1/\epsilon}\cdot{}n)$ edges. By \cref{lem:growth_of_tk}, $\tk[4^{1/\epsilon}n]=O(2^{1/\epsilon}\tn)$, and by \cref{thm:optimal_dissemination}, we make the spanner edges globally known in $\widetilde{O}(2^{1/\epsilon}\cdot{}\tn)$ rounds, and get a $(2k-1)=\epsilon\log{n}-1<\epsilon\log{n}$ approximation for APSP.
\end{proof}

We now desire to show the following.

\secondWeightedAPSPApprox*

Before we do so, we must introduce the well-known concept of skeleton graphs, first observed by \cite{ullman1990high}. A skeleton graph is constructed by sampling each node with probability $1/x$, for some value $x$. The main property is that given $u,v\in{V}$, there will be some shortest path $P_{u,v}$ between $u$ and $v$, there will be a sampled node every $\widetilde{O}(x)$ hops. We use two lemmas from \cite{augustine2020shortest,kuhn2020computing} for skeleton graphs in \hybrid.

\begin{lemma} [Lemma 4.2 in \cite{augustine2020shortest}]
\label{skeleton_prop}
    Let $M\subseteq{}V$ be a subset of nodes of $G=(V,E)$ obtained by sampling each node independently with probability at least $1/x$. There is a constant $\xi>0$, such that for every $u,v\in{V}$ with $\hop(u,v)\geq{}\xi{}x\ln{n}$, there is at least one shortest path $P$ from $u$ to $v$, such that any subpath $Q$ of $P$ with at least $\xi{}x\ln{n}$ nodes contains a node in $M$, \whp
\end{lemma}

The following shows that distances between skeleton nodes in the skeleton graph are the same as distances between them in the original graph, and that we can construct a skeleton graph in \hybrid.

\begin{lemma} [Lemma C.2 in \cite{kuhn2020computing}]
\label{lem:createSkeleton}
    Let $\mathcal{S}=(V_{\mathcal{S}}, E_{\mathcal{S}})$ be a skeleton graph of a connected graph $G$ with $n$ nodes, by sampling each node of $G$ to $V_{\mathcal{S}}$ with probability at least $1/x$. The edges of $\mathcal{S}$ are $E_{\mathcal{S}}=\{\{u,v\} | u,v\in{}V_{\mathcal{S}}, \hop(u,v)\leq{}h\}$ (where $h:=\xi{}x\ln{n}$ is the parameter from \cref{skeleton_prop}), and edge weights $d^h(u,v)$ for $\{u,v\}\in{}E_{\mathcal{S}}$. Then $\mathcal{S}$ is connected and for any $u,v\in{}V_{\mathcal{S}}$ we have $d_G(u,v)=d_{\mathcal{S}}(u,v)$, \whp Further, $\mathcal{S}$ can be computed within $O(h) = \tildeBigO{1/x}$ rounds in \hybrid. 
\end{lemma}

\begin{algorithm}
\caption{$(4\alpha - 1)$-Approximate Weighted APSP} \label{alg:polynomial_approx_apsp}
            Broadcast the identifiers of all the nodes.

            Compute $\tn$. Denote $t = n^{1/(3\alpha + 1)} \cdot (\tn)^{2/(3 + 1/\alpha)}$.
            
            Compute a skeleton graph $G_{\mathcal{S}}=(V_{\mathcal{S}}, E_{\mathcal{S}}, \omega_{\mathcal{S}})$ with sampling probability $1/t$.
            
            Compute a $(2\alpha-1)$-stretch spanner for $G_{\mathcal{S}}$, denoted $K$.

            Broadcast the edges of $K$. Locally compute a $(2\alpha - 1)$ approximation of distances between all nodes in $V_\mathcal{S}$, denoted $\hat{d}$.

            Learn $h = \xi{}t\ln{n}$ neighborhood, where $\xi$ is the constant from \cref{skeleton_prop}.
            
            Each node $v \in V$ denotes by $v_s \in V_\mathcal{S}$ the skeleton node in its $h$-neighborhood with minimal $d^h(v, v_s)$, and broadcasts $v_s$ and $d^h(v, v_s)$.
            
            Node $v$ approximates its distance to any $w\in{V}$ by:
            \[
            \delta(v,w)= \min\{d^h(v,w),d^h(v,v_s)+\hat{d}(v_s,w_s)+d^h(w_s,w)\}
            \]
\end{algorithm}

We now prove \cref{theorem:secondWeightedAPSPApprox}. See \cref{alg:polynomial_approx_apsp} for an overview of our algorithm.

\begin{proof}
    We begin by broadcasting the identifiers of all the nodes in $\tildeBigO{\tn}$ rounds using \cref{thm:optimal_dissemination}. From now on, we can assume we are in \hybrid instead of \hybridzero, and thus we are able to execute algorithms such as \cref{lem:createSkeleton}.
    We proceed to computing $\tk[n]$ in $\tildeBigO{\tk[n]}$ rounds using \cref{thm:tk_computation} and then denote $t = n^{1/(3\alpha + 1)} \cdot (\tn)^{2/(3 + 1/\alpha)}$. Note that throughout the following algorithm, we strive to achieve a round complexity of $\tildeBigO{\alpha \cdot t + \alpha \cdot \tn}$.
    
    Using \cref{lem:createSkeleton}, we compute a skeleton graph $G_{\mathcal{S}}=(V_{\mathcal{S}}, E_{\mathcal{S}}, \omega_{\mathcal{S}})$ with sampling probability $1/t$, in $\tildeBigO{t}$ rounds. Note that $|V_{\mathcal{S}}|=\tilde{\Theta}(n/t)$, \whp. We now create a $(2\alpha - 1)$ spanner of $G_{\mathcal{S}}$ using \cref{thm:spanner_construction}, denoted $K$ -- each round of \cref{thm:spanner_construction} is simulated over $G_{\mathcal{S}}$ using the local edges of $G$, and thus takes $\tildeBigO{t}$ rounds. As \cref{thm:spanner_construction} takes $\tildeBigO{1}$ rounds, our entire simulation takes $\tildeBigO{t}$ rounds.
    
    Due to \cref{thm:spanner_construction}, $K$ has $\tildeBigO{\alpha \cdot |V_{\mathcal{S}}|^{1+1/\alpha}} = \tildeBigO{\alpha \cdot (n/t)^{1+1/\alpha}}$ edges. Set $x = \max\{(n/t)^{1+1/\alpha}, n\}$ and compute $\tk[x]$ in $\tildeBigO{\tk[x]}$ rounds using \cref{thm:tk_computation}. Using \cref{thm:optimal_dissemination}, we can broadcast all of $K$ in $\tildeBigO{\alpha \cdot \tk[x]}$ rounds.
    
    We desire to show that $\tk[x] = O(t + \tn)$. If $x = n$, then trivially $\tk[x] = \tn$. Otherwise, 
    \begin{align*}
        x &= (n/t)^{1+1/\alpha} = n^{1+1/\alpha} \cdot t^{-(1 + 1/\alpha)} \\
        &= n^{1+1/\alpha} \cdot (n^{1/(3\alpha + 1)} \cdot (\tn)^{2/(3 + 1/\alpha)})^{-(1+1/\alpha)} \\
        &= n^{1+1/\alpha - (1+1/\alpha)/(3\alpha + 1)} \cdot (\tn)^{-(2+2/\alpha)/(3+1/\alpha)} \\
        &= n^{((\alpha + 1)/\alpha)\cdot(1-1/(3\alpha + 1))} \cdot (\tn)^{-(2\alpha + 2)/(3\alpha + 1)} \\
        &= n^{((\alpha + 1)/\alpha)\cdot(3\alpha/(3\alpha + 1))} \cdot (\tn)^{-(2\alpha + 2)/(3\alpha + 1)} \\
        &= n^{3\cdot(\alpha + 1)/(3\alpha + 1)} \cdot (\tn)^{-(2\alpha + 2)/(3\alpha + 1)} \\
        &= n^{1 + 2/(3\alpha + 1)} \cdot (\tn)^{-(2\alpha + 2)/(3\alpha + 1)}. 
    \end{align*}
    
    Due to \cref{lem:growth_of_tk}, $\tk[x] = O(\sqrt{x/n} \cdot \tn)$, and so
    \begin{align*}
        \tk[x] &= O(\sqrt{x/n} \cdot \tn) \\
        &= O(\sqrt{n^{2/(3\alpha + 1)} \cdot (\tn)^{-(2\alpha + 2)/(3\alpha + 1)}} \cdot \tn) \\
        &= O(n^{1/(3\alpha + 1)} \cdot (\tn)^{-(\alpha + 1)/(3\alpha + 1)} \cdot \tn) \\
        &= O(n^{1/(3\alpha + 1)} \cdot (\tn)^{2\alpha/(3\alpha + 1)}) \\
        &= O(n^{1/(3\alpha + 1)} \cdot (\tn)^{2/(3 + 1/\alpha)}) \\
        &= O(t). 
    \end{align*}

    Thus, in either case, $\tk[x] = O(t + \tn)$. Therefore using \cref{thm:optimal_dissemination} we can broadcast all the edges in $K$ in $\tildeBigO{\alpha\cdot\tk[x]} = \tildeBigO{\alpha\cdot(t + \tn)}$ rounds. Using this information, each node locally computes a $(2\alpha - 1)$ approximation to the distances in $G_{\mathcal{S}}$.

    Next, every node learns its $h = \xi{}t\ln{n}$ neighborhood in $\tildeBigO{t}$ rounds, where $\xi$ is the constant from \cref{skeleton_prop}. Due to \cref{skeleton_prop}, every node $v$ sees at least one skeleton node in its $h$-neighborhood. Thus, $v$ denotes by $v_s \in V_\mathcal{S}$ the skeleton node in its $h$-neighborhood with minimal $d^h(v, v_s)$, and broadcasts $v_s$ and $d^h(v, v_s)$. This takes $\tildeBigO{\tn}$ rounds, due to \cref{thm:optimal_dissemination}, as every node broadcasts $O(1)$ values.

    Finally, node $v$ approximates its distance to any node $w$ by $\delta(v,w)= \min\{d^h(v,w),d^h(v,v_s)+\hat{d}(v_s,w_s)+d^h(w_s,w)\}$. It remains to show that this is a $(4\alpha-1)$ approximation.

    Let $v,w\in{V}$. If there exists a shortest path between them of less than $h$ hops, then $\delta(v,w)=d^h(v,w)=d(v,w)$. Otherwise, all shortest paths are longer than $h$ hops, and by \cref{skeleton_prop} there exists a skeleton node $s$ on one of them. Further, w.l.o.g., $s$ is in the $h$-neighborhood of $v$. As $s$ sits on a shortest path from $v$ to $w$, then it also holds that $d^h(v, s) = d(v, s)$. Finally, it holds that $d^h(w, w_s) \leq d(w, s)$ -- this is true as either there is a shortest path from $w$ to $s$ with at most $h$ hops, in which case $d^h(w, w_s) \leq d^h(w, s) = d(w, s)$, or, there is a path from $w$ to $s$ with a skeleton node $s'$ on it that is also in the $h$-hop neighborhood of $w$, in which case $d^h(w, w_s) \leq d^h(w, s') \leq d(w, s)$. Using all of these, we show the following.
    \begin{align*}
        \delta(v,w) &= d^h(v,v_s)+\hat{d}(v_s,w_s)+d^h(w_s,w)
    \leq{} d^h(v,v_s) +(2\alpha-1)d(v_s,w_s) + d^h(w_s,w)\\
    &\leq{} d^h(v,v_s) + (2\alpha-1)(d(v_s,v)+d(v,w)+d(w,w_s)) + d^h(w_s,w) \\
    &\leq{} d^h(v,v_s) + (2\alpha-1)(d^h(v_s,v)+d(v,w)+d^h(w,w_s)) + d^h(w_s,w) \\
    &\leq{} (2\alpha-1)d(v,w) + 2\alpha(d^h(v,v_s)+d^h(w_s,w)) \\
    &\leq{} (2\alpha-1)d(v,w) + 2\alpha(d^h(v,s)+d^h(w_s,w)) \\
    &=(2\alpha-1)d(v,w) + 2\alpha(d(v,s)+d^h(w_s,w)) \\
    &\leq{} (2\alpha-1)d(v,w) + 2\alpha(d(v,s)+d(s,w)) \\
    &=(2\alpha-1)d(v,w) + 2\alpha\cdot{}d(v,w) \\
    &=(4\alpha-1)d(v,w)
    \end{align*}

    Note that the approximation never underestimates, i.e.,~$\delta(v, w) \geq d(v, w)$, as it corresponds to actual paths, and thus we are done. 
\end{proof}

\subsection{Approximating Cuts via Learning Spectral Sparsifiers}

We apply our universally optimal broadcasting to different sparsification tools, and in particular cut sparsifiers. We employ the following result from the \congest model. 

\begin{theorem}[Theorem 5 in \cite{koutis2016simple}, rephrased]
\label{thm:congest_cut_sparsifier}
    There is a \congest algorithm, that given a graph $G=(V,E,\omega)$ and any $\epsilon>0$, computes a graph $H=(V, \hat{E}, \hat{\omega})$ such that for any cut $S\subset{}V$ it holds that $(1-\epsilon)\text{cut}_H(S)\leq{}\text{cut}_G(S)\leq{}(1+\epsilon)\text{cut}_H(S)$ and $|\hat{E}|=\tildeBigO{n/\epsilon^2}$ \whp. The algorithm runs in $\tildeBigO{1/\epsilon^2}$ round \whp.
\end{theorem}

We note that we can run this algorithm in \hybridzero with unknown identifiers in $\tildeBigO{\tn+1/\epsilon^2}$ rounds, since we can broadcast all identifiers in $\tildeBigO{\tn}$ rounds by \cref{thm:optimal_dissemination} and proceed to run the algorithm as usual.
Now we easily get \cref{theorem:minCutApprox}.

\minCutApprox*

\begin{proof}
    We run \cref{thm:congest_cut_sparsifier} in $\tildeBigO{\tn+1/\epsilon^2}$ and get a cut sparsifier with $\tildeBigO{n/\epsilon^2}$ edges \whp. By \cref{lem:growth_of_tk} and \cref{thm:optimal_dissemination} we can broadcast the sparsifier in $\tildeBigO{\tn/\epsilon}$ rounds, and each node can compute any cut approximation locally. 
\end{proof}

\section{Estimation of \texorpdfstring{$\tk$}{Tk} on different graph families}
\label{sec:graph_families}
\subsection{Path and cycle graphs}
\label{section:path_graphs}
Arguably, the simplest graph to consider and compute $\tk$ on is the path graph on $n$ nodes $P_n$. By definition, $\tk=\max_{v\in{}V}\{\tk(v)\}$, so it suffices to identify $\argmax_{v\in{}V}\tk(v)$ and compute its $\tk(v)$ value. The $\tk(v)$ value roughly corresponds to the size of the neighborhood of $v$ that has the bandwidth to receive and forward $k$ messages from the rest of the graph in $\tk(v)$ rounds. If the neighborhood is sparser, then $\tk(v)$ will be higher. In a path graph, the corner nodes have the highest $\tk(v)$ value. We formalize this with the following.
\begin{lemma}
\label{thm:tk_nodes_ineq}
    Let $v,w\in{V}$. If for all $t\in{}\mathbb{N}^+$ $|B_t(v)|\leq{}|B_t(w)|$, then $\tk(v)\geq{}\tk(w)$.
\end{lemma}
\begin{proof}
    It follows directly from \cref{def:bcast_quality}. If for any $t$, $|B_t(v)|\leq{}|B_t(w)|$, then the minimal $t$ which satisfies the equation in the definition $|B_t(v)|\geq{}k/t$ is higher for $v$, therefore $\tk(v)\geq{}\tk(w)$.
\end{proof}
\begin{corollary}
    Let $v,w\in{V}$. If $B_{\tk(v)}(v)$ is isomorphic to $B_{\tk(v)}(w)$, then $\tk(w)=\tk(v)$.
\end{corollary}

The condition of \cref{thm:tk_nodes_ineq} obviously holds for the corner nodes $r,l\in{}P_n$ with respect to any other $v\in{P_n}$, since the balls around them only grow to one side. Therefore it suffices to compute the simpler values $\tk(r)$ or $\tk(l)$.

We now compute $\tk(r)$, where $r$ is the rightmost node of the path graph $P_n$: 
\begin{lemma}
    For the path graph on $n$ nodes, $\tk=\begin{cases}
        \Theta{}(\sqrt{k}) & k=O(n^2) \\
        D & k=\Omega{}(n^2)
    \end{cases}$.
\end{lemma}
\begin{proof}
    $|B_t(r)|=t+1$, and by \cref{def:bcast_quality} we get the equation to find the minimum $t\leq{}D=n-1$ such that $t+1\geq{}k/t \implies t^2+t\geq{}k$, which can hold only if $k\leq{}D^2+D=O(n^2)$. Otherwise, the asymptotic solution is $t=\Theta{}(\sqrt{k})$.
\end{proof}

\begin{corollary}
    For a cycle graph $C_n$, the value $\tk$ is the same as for $P_n$, because $|B_t(v)|$ around any $v\in{}V$ is of size $2t+1$, thus the asymptotic solution stays the same. 
\end{corollary}

\subsection{Square grid graphs}

We now turn to computing $\tk$ for square grids in some dimension $d$, that is $n = m^d$. By \cref{thm:tk_nodes_ineq}, it suffices to look at the corner nodes in order to compute $\tk$, since they have smaller neighborhoods compared to other nodes. As a warmup, when $d=2$, $|B_r(v)|=|B_{r-1}(v)|+r+1$, since we add another sub-diagonal with every radius expansion. By induction we get $|B_r(v)|=\sum_{i=0}^{r}i+1=\sum_{i=1}^{r+1}i=\frac{(r+1)(r+2)}{2}=\frac{r^2+3r+2}{2}$. 
By substituting this into \cref{def:bcast_quality}, we get that $\tk=\Theta{}(k^{1/3})$ or $\tk = D$.
We generalize this argument in the following statements.

\begin{lemma}
\label{lemma:square_grid_ring}
    Let $G=(V,E)$ be a $d$-dimensional square grid, with $|V|=n=m^d$. Let $w$ be a corner node, and let $N_r(w)=\{v \mid d(v,w)=r\}$, i.e the set of nodes at distance exactly $r$ from the corner $w$. If $r < n^{1/d}$, then $|N_r(w)|=\binom{r+d-1}{d-1}$.
\end{lemma}
\begin{proof}
The set $N_r(w)$ can be represented by $N_r(w)=\{v\in{}\mathbb{N}^d|\sum_{i=1}^{d}v_i=r\}$, as the distances in a a square grid obey the $L_1$ metric, and without loss of generality, the corner node $w$ corresponds to the $\vec{0}$ vector in $\mathbb{N}^d$. Therefore the number of elements in $N_r(w)$ is  given by $\binom{r+d-1}{d-1}$, as shown in \cite{wikipedia_combi}. We require that $r<n^{1/d}$ so the number of weak compositions describes correctly the ring size. This way, no component of the composition exceeds the edge length $m=n^{1/d}$, so each composition is a unique and existing node in the grid graph.
\end{proof}

\begin{lemma}
\label{lemma:square_grid_neighborhood}
    If $r < n^{1/d}$, then by the conditions of \cref{lemma:square_grid_ring}, $|B_r(v)|=|N_0(w)\cup{}\dots\cup{}N_r(w)|=\binom{r+d}{d}$.
\end{lemma}
\begin{proof}
    Observe the following transitions.
\[
|B_r(v)|=|N_0(w)\cup{}\dots\cup{}N_r(w)|=\sum_{i=0}^{r}|N_i(w)|=\sum_{i=0}^{r}\binom{i+d-1}{d-1}\overset{m=d-1}{=}\sum_{i=0}^{r}\binom{i+m}{m}
\]
\[
\overset{j=i+m}{=}\sum_{j=m}^{m+r}\binom{j}{m}\overset{\text{Pascal's Triangle}}{=}\sum_{j=m}^{m+r}[\binom{j+1}{m+1}-\binom{j}{m+1}]\overset{\text{Telescoping Sum}}{=}\binom{m+r+1}{m+1}=\binom{r+d}{d}
\]
\end{proof}

This leads us to the following theorem for $d$-dimensional square grids.

\gridGraphsTk*

\begin{proof}
    From \cref{lemma:square_grid_neighborhood}, we note that $|B_r(v)|$ for a corner node $v$ is a polynomial of degree $d$ in $r$, since $\binom{r+d}{d}=\frac{(r+d)\cdot{}\dots{}\cdot{}(r+1)}{d!}$.

    The requirement that $r<m=n^{1/d}=D/d$ means that the approximate solution is within a $d$ factor from $D$, so our computation is approximately correct up to a $d$ factor.
    
    Thus, the equation $|B_r(v)|\geq{}k/r$ of \cref{def:bcast_quality} becomes $r\cdot{}\binom{r+d}{d}\geq{}k$, which is a polynomial of degree $d+1$, with constant term $a_0=\Theta{}(k)$. The diameter of a $d$-dimensional grid is $dk=dn^{1/d}$ where $k$ is the length of a side. We are looking for the minimal solution $r$ that is less than $D=dk=dn^{1/d}$. Therefore considering that $r$ also has to be at most $D$, the asymptotic solution is as follows.
    \[
    \tk=\begin{cases}
        \Theta{}(k^{1/(d+1)}) & k=O(D^{d+1})=O(n^{(d+1)/d}) \\
        O(D)=O(dn^{1/d}) & k=\Omega(D^{d+1})=\Omega(n^{(d+1)/d})
    \end{cases}
    \]
\end{proof}

\bibliographystyle{plainurl}
\bibliography{hybrid}

\newpage
\appendix

\section{Preliminary Algorithms in \hybrid}

\subsection{Deterministic Virtual Tree Construction}
\label{appendix:deterministic_ID}

We show how to adapt the deterministic overlay network construction of \cite{gmyr2017distributed} to \hybridzero. This allows us to deterministically construct a virtual tree with constant maximal degree and polylogarithmic depth in a polylogarithmic number of rounds.

The following theorem follows from \cite{gmyr2017distributed}.

\begin{theorem}[Theorem 2 in \cite{gmyr2017distributed}, rephrased]
\label{thm:fakeModelTree}
    Given a connected graph $G$ of $n$ nodes and polylogarithmic maximal degree, there is an algorithm that constructs a constant degree virtual tree of depth $O(\log{n})$ that contains all nodes of $G$ and that is rooted at the node with the highest identifier. The algorithm takes $O(\log^2{n})$ rounds in \hybridzero.
\end{theorem}

We show how to remove the constraint that the graph $G$ originally has polylogarithmic maximal degree, to show the following.

\treeWithoutIDs*

\begin{proof}
    Every node $v \in V$ begins by denoting its neighbor with highest identifier with $m(v)$. If $v$ has an identifier which is higher than all its neighbors, then it sets $m(v) = v$. Node $v$ sends a message to $m(v)$ stating that it desires to join it. 
    
    Observe the directed graph $G_1 = (V, E_1)$ where there is an edge from each $v$ to $m(v)$, if $v \neq m(v)$. Every node in this graph has out-degree at most $1$, yet unbounded in-degree. Let $u \in V$ be some node in $G_1$ with in-degree greater than $1$. Denote by $v_1, \dots, v_x$, for some $x > 1$, the nodes with edges towards $u$ in $G$. For each $v_i$, node $u$ sends a message to $v_i$ to adjust its edge so that instead of it pointing to $u$, it points to $v_{i+1}$. To node $v_x$, node $u$ does not send a message (i.e. the edge still points to $u$). Denote this new graph by $G_1'$.

    Notice that in $G_1'$, every node has out-degree at most $1$ and in-degree at most $2$. First, observe that in $G_1$ each node has out-degree at most $1$, and this edge might have only been replaced by one edge in $G_1'$, so each node still has out-degree at most $1$. As for the in-degree, any node $u\in{}G_1'$ has at most two edges coming into it. Denote by $v$ the node that $u$ points to in $G_1$. In $G_1'$, the only edges coming into $u$ can be the at most one edge which came into it in $G_1$, and one edge which $v$ told some other node $w$ to create towards $u$.

    Now, drop the directions of the edges in $G_1'$ to receive an undirected graph with constant maximal degree. Execute \cref{thm:fakeModelTree} so that there is a tree of polylogarithmic depth and constant maximal degree for every connected component in $G_1'$. This takes $\tildeBigO{1}$ rounds.

    Observe the connected components in $G_1'$. Each connected component now has a virtual tree spanning it, and every node knows the identifiers of its neighbors in the virtual tree which it includes it. In each connected component, the nodes use the tree to compute, in $\tildeBigO{1}$ rounds, the node with maximal identifier in the component. We now create a new graph $G_2 = (V_1, E_2)$ as follows. Each node in $G_2$ corresponds to one connected component in $G_1'$ and has the identifier of the node with maximal identifier in that component. For any node $v \in V_1$ in $G_2$, denote its corresponding connected component in $G_1'$ by $c(v)$. The nodes in $c(v)$ compute the node $u \in V_1$ with maximal identifier such that there are nodes $x \in c(v), y \in c(u)$ where $x, y$ are neighbors in $G$ (or set $u = v$ if the identifier of $v$ is greater than all of the identifiers of the clusters neighboring it). To compute $u$, every node in $G$ tells its neighbors the identifier of the cluster to which it now belongs. Now, all nodes in $c(v)$ know the identifiers of the clusters to which their neighbors belong, and thus the nodes $c(v)$ can compute $u$ using one aggregation over the virtual tree that connects the nodes $c(v)$.

    As such, we have finished constructing $G_2$. Clearly, $G_2$ has at most half as many nodes as $G_1$, as every node in $G_1$ was merged with at least one other node in the clustering phase. Thus, we can repeat the above for $O(\log n)$ iterations, each time simulating the cluster nodes using the virtual trees which span them, and uniting trees whenever combining clusters.  Whenever we create a new cluster, it has a tree of at most polylogarithmic maximal degree and depth, and thus we can rerun \cref{thm:fakeModelTree} inside the cluster in order to make sure its maximal degree is constant and has at most $O(\log n)$ depth.
\end{proof}

\subsection{Basic Properties of \texorpdfstring{$\tk$}{Tk}}
\label{appendix:basic_properties_of_tk}

We provide the proofs for the statements in \cref{sec:optimal_broadcast;subsec:basic_properties}.

\tkComputation*

\begin{proof}
    By \cref{def:bcast_quality}, if a node $v$ knows the value $k$ and $B_{\tk}(v)$, it can compute $\tk(v)$. We use that fact to incrementally explore larger neighborhoods and stop when $v$ knows $\tk(v)$.
    First, using \cref{lem:hybrid_zero_aggregate}, all nodes will aggregate the number of tokens they possess using the sum function. This takes $\tildeBigO{1}$ rounds and makes $k$ globally known.

    We then proceed in $\tk$ iterations. Each iteration starts by $v$ using the local network to learn its neighborhood to one more hop -- that is, in iteration $t$, every node $v$ knows $B_t(v)$. Using this information, and as $k$ is globally known, $v$ can determine whether $\tk(v)=t$. Then, we compute how many nodes in the graph have $\tk(u) \leq t$, using one aggregation in $\tildeBigO{1}$ rounds. If we know that all $n$ nodes have $\tk(u) \leq t$, then we halt and set and can compute $\tk=t$. 
\end{proof}

\growthOfTk*

\begin{proof}
    If $6 \sqrt{\alpha} \cdot \tk \geq D$, then by definition $\tk[\alpha{}k] \leq D \leq 6 \sqrt{\alpha} \cdot \tk$. Thus, assume that $6 \sqrt{\alpha} \cdot \tk < D$. Let $v$ be some node. There must exist $u$ such that the hop-distance between $v$ and $u$ is at least $D/2 > 3 \sqrt{\alpha} \cdot \tk$, otherwise the diameter of the graph is less than $D$. Denote by $P = (v_1 = v, \dots, v_\ell = u)$ a shortest path between $v$ and $u$ of length $\ell > 3 \sqrt{\alpha} \cdot \tk$. 
    
    Observe any two nodes on the path which are $3\tk$ edges apart on the path -- i.e.,~nodes $v_i, v_{i + 3\tk}$. It holds that $B_{\tk}(v_i) \cap B_{\tk}(v_{i + 3\tk}) = \emptyset$. Assume for the sake of contradiction that $B_{\tk}(v_i) \cap B_{\tk}(v_{i + 3\tk}) \neq \emptyset$ and take $w \in B_{\tk}(v_i) \cap B_{\tk}(v_{i + 3\tk})$. It holds that $d(v_i, w) \leq \tk$ and $d(w, v_{i + 3\tk}) \leq \tk$, and thus $d(v_i, v_{i + 3\tk}) \leq 2\tk$. However, as $P$ is a shortest path, then $d(v_i, v_{i + 3\tk}) = 3\tk > 2\tk$, where the last inequality holds since $\tk \geq 1$, and so we arrive at a contradiction.

    Therefore, we get that 
    \begin{align*}
        |B_{3 \sqrt{\alpha} \cdot \tk}(v)| &\geq \sum_{i \in [\sqrt{\alpha}]} |B_{\tk}(v_{1 + i \cdot 3\tk})| \\
        \geq \sqrt{\alpha}\cdot{}\frac{k}{\tk} 
        &= \frac{\alpha{}k}{\sqrt{\alpha}\tk} 
        \geq{} \frac{\alpha{}k}{3\sqrt{\alpha}\tk}.
    \end{align*}
    
    Thus, by definition, $\tk[\alpha{}k] \leq 3\sqrt{\alpha} \cdot{} \tk$, and we are done. 
\end{proof}

\boundingTk*

\begin{proof}
    For the first inequality, we note that for all $v\in{}V$ it holds that $|B_{\tk}(v)|\geq{}k/\tk$ since $\tk\neq{}D$, and therefore $\frac{n}{|B_{\tk}(v)|}\leq{}\frac{n\tk}{k}$. Thus, there can be at most $\frac{n\tk}{k}$ disjoint balls of radius $\tk$. Given a shortest path $P_{u,v}$, it visits at most $2\tk+1$ nodes in each such ball. Otherwise the diameter of this ball would be larger, since the path between the first and the last $s,t\in{}B_{\tk}(v)$ is also a shortest path, and the diameter of the ball is $2\tk$. Therefore, as it is always true that $\tk \geq 1$, we get that $D\leq{}\frac{n\tk}{k}\cdot{}(2\tk+1)\leq{}3\frac{n\tk^2}{k}$ and we get $\tk\geq{}\sqrt{\frac{D}{3n}k}$.

    For the second inequality, $\tk\leq{}D$ is holds by \cref{def:bcast_quality}. 
    Let $v\in{V}$. If $\sqrt{k}\leq{}n$, then $|B_{\sqrt{k}}(v)|\geq{}\sqrt{k}=k/\sqrt{k}$ because we assume $G$ is connected and with every additional hop we explore, we expand by at least one node. Therefore, $\tk(v)\leq{}\sqrt{k}$ by \cref{def:bcast_quality}. We get that $\tk(v)\leq{}\sqrt{k}$ for all $v\in{V}$ and thus $\tk\leq{}\sqrt{k}$.
    If $\sqrt{k}>n$, then $B_{\sqrt{k}}(v)=V$ and so $D \leq \sqrt{k}$. Thus, as $\tk \leq D$, it holds that $\tk \leq \sqrt{k}$. 
\end{proof}

\end{document}